\documentclass[11pt,a4paper]{article}
\usepackage{amssymb}
\usepackage{amsmath, amscd}
\usepackage{amsfonts}
\usepackage{graphicx}
\usepackage{a4wide}
\usepackage{microtype}
\usepackage{bbm}
\usepackage{slashed}
\usepackage{times}
\usepackage{amsthm}
\usepackage{enumerate}
\usepackage{mathrsfs}
\usepackage{hyperref}
\usepackage{color}

\usepackage[latin1]{inputenc}
\usepackage{url}
\usepackage[super,compress]{cite}

\definecolor{darkred}{rgb}{0.8,0.1,0.1}
\hypersetup{
     colorlinks=true,         
     linkcolor=darkred,
     citecolor=blue,
}

\usepackage{comment}
\usepackage[margin=2.35cm]{geometry}
\usepackage[all,cmtip]{xy}

\usepackage{marvosym}

\theoremstyle{plain}
\newtheorem{theo}{Theorem}[section]
\newtheorem{lem}[theo]{Lemma}
\newtheorem{propo}[theo]{Proposition}

\theoremstyle{definition}
\newtheorem{defi}[theo]{Definition}
\newtheorem{ex}[theo]{Example}

\newtheorem{rem}[theo]{Remark}

\numberwithin{equation}{section}

\newcommand{\supp}{\textrm{supp} }
\newcommand{\wick}[1]{:\!{#1}\!:}

\title{Quantum Field Theory on Curved Backgrounds - A Primer}

\author{%
Marco Benini$^{1,3,a}$, Claudio Dappiaggi$^{1,b}$ and Thomas-Paul Hack$^{2,c}$\vspace{4mm}\\
{\small $^1$ Dipartimento di Fisica}\\ 
{\small Universit{\`a} di Pavia \& INFN, sezione di Pavia - Via Bassi 6, I-27100 Pavia, Italy.}\vspace{2mm}\\
{\small $^2$ Dipartimento di Matematica}\\
{\small Universit\`a di Genova - Via Dodecaneso 35, I-16146 Genova, Italy.}\vspace{4mm}
\\
{\small $^3$ II. Institut f\"ur Theoretische Physik}\\
{\small Universit\"at Hamburg - Luruper Chausse 149, D-22761 Hamburg, Germany.}\vspace{4mm}
\\
 {\small  ~$^a$ marco.benini@pv.infn.it~,~$^b$ claudio.dappiaggi@unipv.it~,~$^c$ hack@dima.unige.it}
 }

\date{\today}

\begin{document}

\maketitle

\begin{abstract}
Goal of this review is to introduce the algebraic approach to quantum field theory on curved backgrounds. Based on a set of axioms, first written down by Haag and Kastler, this method consists of a two-step procedure. In the first one, it is assigned to a physical system a suitable algebra of observables, which is meant to encode all algebraic relations among observables, such as commutation relations. In the second step, one must select an algebraic state in order to recover the standard Hilbert space interpretation of a quantum system. As quantum field theories possess infinitely many degrees of freedom, many unitarily inequivalent Hilbert space representations exist and the power of such approach is the ability to treat them all in a coherent manner. We will discuss in detail the algebraic approach for free fields in order to give to the reader all necessary information to deal with the recent literature, which focuses on the applications to specific problems, mostly in cosmology. 

\end{abstract}

\section{Introduction}	

The twentieth century will be forever remembered in theoretical and mathematical physics for the formulation of general relativity and of quantum field theory. The first one has revolutionized our understanding of the gravitational force and of the nature of space and time, superseding the Galilean notion of absolute space with that of a Lorentzian manifold. Thereon the role and the effect of gravity is encoded in the metric field whose intertwinement with matter is ruled by Einstein's equations via the stress-energy tensor. The second one has completely changed our view of the physical description of the basic constituents of matter, as we observe it. Quantum electrodynamics and the so-called standard model of particles have been experimentally verified to an outstanding degree of precision and allowed us to have an almost fully satisfactory and unified description of the electro-weak forces. It is striking that, to the inherent quantum nature of the description of elementary particles, it is opposed the structure of a classical theory, proper of general relativity. It is almost unanimously accepted that this quandary can be solved by developing a quantum version of Einstein's theory, which is often dubbed as quantum gravity. Yet, despite countless efforts, a quantum theory of the gravitational interaction remains elusive at best.

Nonetheless, the quest to finding a quantum theory of gravity often lead the community to neglect the existence of an intermediate regime which has acquired only recently a renewed relevance. As a matter fact, both quantum electrodynamics, the standard model of particles, as well as all the quantum theories thought in the standard undergraduate courses, are formulated under the assumption that the underlying background is Minkowski spacetime. This is $\mathbb{R}^4$ endowed with the metric tensor $\eta$, which in Cartesian coordinates reads as the matrix diag$(-1,1,1,1)$. Yet, if we take into account the lesson of general relativity, Minkowski spacetime is a rather special case of a Lorentzian manifold, which solves Einstein's equations only in absence of matter. At the same time, changing the structure of the background might look at first glance as an hazard. On the one hand, unless in very extreme conditions such as in a neighbourhood of a black hole, the gravitational effects are negligible and no direct evidence of the necessity to modify the spacetime metric has emerged in the experiments based at the high energy accelerators. On the other hand, the high degree of symmetry of the Minkowski metric, encoded in the Poincar\'e group, makes the formulation of any quantum field theory technically much easier then the counterpart on an arbitrary Lorentzian background. Basic tools such as Fourier transform have no meaning nor counterpart as soon as we consider a manifold endowed with a non trivial metric.

For these reasons, the development of the so-called {\em quantum field theory on curved backgrounds} looked often more as a mere intellectual curiosity rather than a stringent necessity. This perception has drastically changed in the last twenty years or so, particularly thanks to the advances in our understanding of the Universe as we observe it. Cosmology has become one of the leading branches of theoretical physics, mostly because it offers the possibility to interface the models built so far with a plethora of experimental data. Since the formulation of theories like inflation, it has become increasingly clear not only that the matter description via quantum fields has to be taken into account, but also that, due to the large scales involved (both spatially and temporally), the presence of a non trivial background must be accounted for. Although gravity is still weak enough to justify a classical description of such force, the effects of a non Minkowskian metric on the quantum fields are no longer negligible.

At this stage, one could follow two different philosophies: either look for an adaptation of the standard techniques to a few special cases of interest, such as cosmological spacetimes, or seek a formulation, based on suitable first principles, which can be applied to the largest possible class of backgrounds. In this review we shall adhere to the latter point of view and we shall present the so-called algebraic approach to quantum field theory on curved backgrounds. Based on a set of axioms, first written down by Haag and Kastler\cite{Haag:1963dh} for quantum theories on Minkowski spacetime and adapted to a curved background setting by Dimock\cite{Dimock}, the algebraic approach  characterizes the quantization of any field theory as a two-step procedure. In the first, one assigns to a physical system a suitable $*$-algebra $\mathcal{A}$ of observables which encodes the relevant physical properties of isotony, locality and covariance, as well as the algebraic relations among observables such as the commutation relations. In particular the last two entail that both the compatibility with background isometries and the property, that observables which are spacelike separated must commute, are automatically built in $\mathcal{A}$. The second step consists of selecting a so-called algebraic state $\omega$ that is a continuous linear and positive functional on $\mathcal{A}$, which, via the GNS theorem\cite{Fredenhagen}, allows us to recover the interpretation of the elements of $\mathcal{A}$ as linear operators on a suitable Hilbert space.

Goal of this paper will be chiefly to acquaint the reader with these two steps. We will focus particularly on free field theories and we will show that a full-fledged characterization of the space of solutions of the partial differential equations ruling their dynamics allows for an explicit construction of a $*$-algebra of observables which encodes all the desired structural properties. Subsequently we will introduce the concept of an algebraic state. At this stage one of the key peculiarities of quantum field theory on curved backgrounds will emerge. In a free quantum field theory on Minkowski spacetime, under the assumption to work at vanishing temperature, there exists a unique ground state which can be singled out thanks to the Poincar\'e group\cite{Haag:1992hx}. The relevance of the so-called Poincar\'e vacuum is enormous, being at the heart of the perturbative methods out of which interactions are treated. On the contrary, whenever the underlying manifold is no longer trivial, the lack of a sufficiently big group of isometries of the metric forbids the existence of a natural vacuum state. We will show, moreover, that most of the algebraic states are not even physically acceptable since they entail that the basic objects of perturbation theory, the Wick polynomials, cannot be well-defined. We shall identify, therefore, a distinguished subclass of states, called {\em Hadamard states}, which avoid these pathologies and are unanimously recognized as being the only physically significant ones.

The goal of the review is not to update the reader on all the achievements of the algebraic formulation of quantum field theory on curved backgrounds, but, rather to introduce her/him all preliminary tools which are necessary to read the literature discussing the latest developments. From this point of view, it is worth spending a few words on the topics beyond the scopes of this review, in which algebraic quantum field theory on curved backgrounds plays a key role. Probably one of the more recent frameworks, where it has been applied extensively, is cosmology. The main questions addressed have been the explicit construction of physically acceptable quantum states\cite{Dappiaggi:2007mx, Dappiaggi:2008dk} and the construction of solutions of the semiclassical Einstein's equations on homogeneous and isotropic backgrounds\cite{Dappiaggi:2008mm, Eltzner:2010nx, Hack:2010iw, Pinamonti:2010is}. The success of these early analyses prompted additional investigations ranging from the power spectrum of the cosmic microwave background\cite{Pinamonti:2013zba}, to spin $2$ fields\cite{Fewster:2012bj}, to the quantization of perturbations in inflationary models\cite{Eltzner:2013soa}, to particle production on cosmological spacetimes\cite{Degner:2009rq}, to non-standard equations of state, such as the Chaplygin gas\cite{Zschoche:2013ola} and to an understanding of the standard cosmological model from first principles \cite{DHP3}. It is noteworthy that the algebraic approach appears to be sufficiently versatile to be  applicable efficiently to the study of quantum fields on homogeneous, but anisotropic backgrounds\cite{Avetisyan:2012zn}.

As we mentioned above, the presence of a non trivial background is expected to play a major role particularly in the neighbourhood of regions with a strong gravitational field such as black holes. This is a framework in which, historically, algebraic quantum field theory has been applied very successfully in order to derive mathematically rigorous results concerning the structure of quantum field theories in these backgrounds\cite{Kay:1988mu}, Hawking radiation\cite{Fredenhagen:1989kr, Moretti:2010qd} and the construction of the Unruh state\cite{Dappiaggi:2009fx}.

Additionally, the algebraic approach to quantum field theory on curved backgrounds has been used to investigate several structural and formal properties ranging from quantum energy inequalities\cite{Fewster:2001js, Fewster:2003ey, Fewster:2006kt, Fewster:2007rh, Fewster:2012yh}, to the notion of thermal equilibrium states\cite{Schlemmer:2008dk}, to the spin-statistics\cite{Verch:2001bv} and Reeh-Schlieder theorem\cite{Dappiaggi:2011ms, Sanders:2008gs}, to the principle of general local covariance\cite{Brunetti:2001dx} and of dynamical locality\cite{Fewster:2011pe, Fewster:2011pn}, to superselection sectors\cite{Brunetti:2005gw, Brunetti:2008hq} up to the role of renormalization\cite{Brunetti:1999jn, Brunetti:2009qc}. It is noteworthy that recently, in Ref. \citen{Hollands:2008vx}, a new axiomatic approach to quantum field theory on curved backgrounds has been proposed and it is currently investigated. Additionally the recent developments in perturbative algebraic quantum field theory on curved backgrounds inspired an interesting novel structural approach to classical field theory\cite{Brunetti:2012ar}.

Before outlining the content of the various sections of this paper, we stress that this is certainly not the first review on quantum field theory on curved backgrounds. Other relevant sources, whose presentation is often complementary to ours are Refs. \citen{Bar:2009zzb, BG2, Brunetti:2009pn, Fewster, Fredenhagen}, as well as the book written by Wald\cite{Wald2}.

This review is organized as follows: In Section 2 we discuss the key geometrical concepts which lie at the heart of the construction of quantum field theory on curved backgrounds, particularly the notion of globally hyperbolic spacetimes. In Section 3, we focus instead on the classical description of linear free field theories, whose dynamics is ruled by hyperbolic partial differential equations on vector bundles. In particular we will be interested in those equations which admit advanced and retarded Green operators. Three prototype examples are developed in detail, namely the scalar, the Majorana and the Proca fields. Green operators turn out to be relevant for implementing the dynamics at the level of the quantum theory. This is is discussed at the beginning of Section 4 adopting the algebraic framework. In particular we show how it is possible to define an algebra of fields for a Bosonic or Fermionic theory starting from the classical data. To complete the discussion, we introduce in Section 4 the notion of an algebraic state on the resulting quantum field theory and we prove the GNS theorem which allows us to interpret the said algebra in terms of linear operators on a suitable Hilbert space. Hence we recover the standard interpretation of a quantum theory. Eventually we introduce some basic tools proper of microlocal analysis in order to characterize physically sensible states by means of the so-called Hadamard condition. A few examples known in the literature are sketched.

\section{Globally hyperbolic spacetimes}

Goal of this section is to make the reader acquainted with the main geometric tools, one needs to formulate and to study all quantum field theories on curved backgrounds. In order to keep under control the length of this review, we assume that the reader is familiar with the basic notions of differential geometry. We will follow the notations and conventions of Refs. \citen{BGP,Wald}. Hence we shall start from the building block on which the whole theory of quantum fields on curved backgrounds is based:

\begin{defi}
A {\it spacetime} $M$ is a four-dimensional, orientable, differentiable, Hausdorff, second countable manifold endowed with a smooth Lorentzian metric $g$ of signature $(-,+,+,+)$. 
\end{defi}

We stress that the requirement on the dimensionality of $M$ is only based on our desire to describe the natural generalizations to curved backgrounds of the field theories, on which the current models of particle physics are based. Most of the concepts we will review in this section can be slavishly adapted to any dimension. The Lorentzian character of the metric $g$ plays an important distinguishing role for the pair $(M,g)$ since we can define two additional concepts: {\it time orientability} and {\it causal structure}. As a matter of fact, these are generalizations of the very same concepts, one can formulate within the theory of special relativity. In the geometric language we use, the latter can be recovered by fixing $M$ as $\mathbb{R}^4$ and $g$ as the Minkowski metric $\eta$. If we consider now any two events in $(\mathbb{R}^4,\eta)$, marked as two distinct points $p,q\in\mathbb{R}^4$, there exists a geometric method to say whether $p$ lies in the future of $q$ (or viceversa) and whether $p$ and $q$ are causally connected. To be precise, if we fix the standard flat global coordinates $(t,x,y,z)$, we can construct two distinguished spacetime regions $J_{\mathbb{R}^4}^+(p)$ and $J_{\mathbb{R}^4}^-(p)$ as the subsets of $\mathbb{R}^4$, constituted by those points $q\in M$ such that the vector $v_{pq}$ connecting $p$ to $q$ has strictly positive (resp. strictly negative) time-component and length $\eta(v_{pq},v_{pq})\leq 0$. In particular if the equality holds, we say that $v_{pq}$ is {\it lightlike}, otherwise we call it a {\it timelike} vector. In other words, with respect to $p$, we have divided $\mathbb{R}^4$ in three regions, namely, besides $J_{\mathbb{R}^4}^+(p)$ and $J_{\mathbb{R}^4}^-(p)$, there exists also the collection of points $p$ for which $v_{pq}$ is {\it spacelike}, or, equivalently, for which $p$ and $q$ are causally separated, that is $\eta(v_{pq},v_{pq})>0$. Notice that the point $p$ itself plays a distinguished and separate role and it is here conventionally assumed to include it in both $J_{\mathbb{R}^4}^+(p)$ and $J_{\mathbb{R}^4}^-(p)$. The regions $J_{\mathbb{R}^4}^+(p)$ and $J_{\mathbb{R}^4}^-(p)$ are called the {\it future (resp. past) light cone}, stemming from $p$.

On a generic background $M$, the above division cannot be applied slavishly, since, in general, there exists neither a notion of future and past nor that of a vector joining two distinct events. Yet, it is possible to circumvent this obstruction by recalling that $T_pM$, the tangent space at any point $p\in M$, is isomorphic to $\mathbb{R}^4$ as a vector space.  Hence, it is possible to attach to each $v_p\in T_pM$ the etiquette of timelike, spacelike or lightlike vector depending whether its length $g(v_p,v_p)$ is smaller, greater or equal to $0$ respectively. In this way, as in Minkowski spacetime, we divide $T_pM$ into two regions, the set of spacelike vectors, and the two-folded light cone stemming from $0$, the origin of the vector space $T_pM$. For each point $p\in M$ we have the freedom to designate each of the folds of the light cone of $T_pM$ as the set of future-directed and of past-directed vectors respectively. If we can smoothly specify at each point which one of the two cones is the future one, we say that $(M,g)$ is {\it time orientable}\cite{Wald}. This is equivalent to the existence of a global vector field on $M$ which is timelike everywhere. Henceforth we will only consider pairs $(M,g)$ enjoying this property and, moreover, we shall assume that a time orientation has been fixed. This allows to introduce $J_M^\pm$ even when $(M,g)$ is not isometric to $(\mathbb{R}^4,\eta)$. To be precise, in the first place one has to define {\it timelike, lightlike and spacelike curves}: A piecewise smooth curve $\gamma:I\to M$, $I\in[0,1]$, is timelike (lightlike or spacelike) if, for every $t\in I$, the vector tangent to the curve at $\gamma(t)$ is timelike (respectively lightlike or spacelike). A curve is called {\it causal} if it is nowhere spacelike. For causal curves, one can also specify the direction according to the time orientation of $(M,g)$: A causal curve $\gamma:I\to M$ is {\it future- (past-) directed} if for each $t\in I$ each vector tangent to $\gamma$ at $\gamma(t)$ lies in the future (respectively past) fold of $T_{\gamma(t)}M$. Given these preliminaries, for any $p\in M$ we call {\it causal future} of $p$ the set $J_M^+(p)$ of points $q\in M$ which can be reached by a future-directed causal curve stemming from $p$. Replacing future-directed causal curves with past-directed ones, we define also the {\it causal past} of $p$, $J_M^-(p)$. Notice that, if we allow only timelike (in place of causal) curves, we replace $J_M^\pm(p)$ with $I_M^\pm(p)$, namely the {\it chronological future (+) and past (-)} of $p$. As in Minkowski spacetime, we assume as a convention that $p$ lies in both $J_M^+(p)$ and $J_M^-(p)$ but neither in $I_M^+(p)$ nor in $I_M^-(p)$. Moreover to any subset $\Omega\subset M$ we can associate $J_M^\pm(\Omega)\doteq\bigcup\{J_M^\pm (p):p\in\Omega\}$ and $I_M^\pm(\Omega)\doteq\bigcup\{I_M^\pm (p):p\in\Omega\}$.

At this stage, the collection of all time orientable spacetimes $(M,g)$ is far too big for our purposes. On the one hand, our goal is to construct quantum field theories on curved backgrounds; hence, as a starting point, we have to make sure that their dynamics can be meaningfully discussed in terms of an initial value problem. On the other hand the causal structure defined on $(M,g)$ may lead in some cases to scenarios which are pathological from a physical point of view, for example due to the appearance of future- (or past-) directed closed causal curves. The prime example of our concern is the often mentioned G\"odel spacetime, for which $M=\mathbb{R}^4$ whereas the line element in the standard coordinates $(t,x,y,z)$ reads\cite{Pfarr}:
$$ds^2=-\left(dt+e^{2ky}dx\right)^2+dy^2+\frac{e^{4ky}}{2}dx^2+dz^2,$$
where $k\in\mathbb{R}$ is a constant. If we introduce the following coordinate transformations:
\begin{gather*}
e^{2ky}=\cosh(2kr)+\sinh(2kr)\cos\varphi,\quad\sqrt{2}kxe^{2ky}=\sinh(2kr)\sin\varphi,\\
\frac{kt}{\sqrt 2}=\frac{k t^\prime}{\sqrt 2}-\frac{\varphi}{2}+\arctan\left(e^{-2kr}\tan\frac{\varphi}{2}\right),
\end{gather*}
where $|k(t-t^\prime)|<\frac{\pi}{\sqrt 2}$, $r\in [0,\infty)$ and $\varphi\in [0,2\pi)$, then 
$$ds^2=-dt^{\prime 2}+dr^2+dz^2-\frac{\sqrt{8}}{k}\sinh^2(kr)d\varphi dt^\prime+\frac{1}{k^2}\left(\sinh^2(kr)-\sinh^4(kr)\right)d\varphi^2.$$
One can directly check that any curve, for which both $t^\prime$ and $z$ are arbitrarily fixed, whereas $r$ is set to be equal to a fixed value larger or equal to $r_0=(1/k)\mathrm{ln}(1+\sqrt{2})$, is closed and causal (lightlike for $r=r_0$, timelike otherwise).

The constraints on $M$ and $g$, implying that these pathological situations are avoided, have been studied in great detail, particularly in the sixties and in the seventies. The solution relies on the notion of a {\it causally convex open set} $\Omega$ of $M$, that is, for a given spacetime $M$, $\Omega\subseteq M$ is open and no future- (or equivalently past-) directed causal curve has a disconnected intersection with $U$. Consequently, $M$ is called a {\it strongly causal spacetime}, provided the existence of arbitrarily small causally convex neighbourhoods of each point. More precisely, the condition is the following: for each $p\in M$ and for each neighbourhood $U\subseteq M$ of $p$, there exists a causally convex neighbourhood $V\subseteq M$ of $p$ such that $V\subseteq U$.  Such a requirement prevents $M$ from admitting future-directed closed causal curves since any such curve would be included in any neighbourhood of any point in its image. The notion of strong causality entails several additional interesting mathematical properties, we will not discuss in details. An interested reader should refer to Ref. \citen{BEE} for a thorough analysis.

From a physical point of view, we are interested in those spacetimes which allow to set a well-posed initial value problem for hyperbolic partial differential equations, such as the scalar D'Alambert wave equation, to quote the simplest, yet most important example. In particular we need to ensure that the spacetime we consider possesses at least one distinguished codimension $1$ hypersurface on which we can assign the initial data needed to construct a solution of such an equation. Hence a concept slightly stronger than that of strong causality is required and this goes under the name of {\it global hyperbolicity}, see Ref. \citen{Wald}, Section 8. 

We begin by introducing an {\it achronal} subset of a spacetime $M$, namely a subset $\Sigma$ such that $I_M^+(\Sigma)\cap\Sigma=\emptyset$. Subsequently we associate to $\Sigma$ its future {\it domain of dependence} as the collection $D_M^+(\Sigma)$ of points $p\in M$ such that every past-inextensible\footnote{A curve is called {\it inextensible} when every extension of the curve coincides with it. {\it Past-inextensibility} is a slightly weaker condition which constrains a causal curve to coincide with each extension by a causal curve in the past.} causal curve passing through $p$ intersects $\Sigma$. Equivalently one defines $D^-(\Sigma)$ as the past domain of dependence. {\it Cauchy hypersurfaces} are defined as closed achronal subsets $\Sigma$ of a spacetime $M$ such that their domain of dependence $D_M(\Sigma)=D_M^+(\Sigma)\cup D_M^-(\Sigma)$ coincides with $M$. It is noteworthy that the term ``hypersurface'' is not used by chance, since one can prove that $\Sigma$ is a three-dimensional, embedded, $C^0$ submanifold of $M$, {\it cf.} Theorem 8.3.1 in Ref. \citen{Wald}. The relevance of Cauchy hypersurfaces is related to the definition of globally hyperbolic spacetimes.

\begin{defi}\label{globhyp}
A spacetime $M$ is called {\it globally hyperbolic} if and only if there exists a Cauchy hypersurface.
\end{defi}

Globally hyperbolic spacetimes represent thus the canonical class of backgrounds on which quantum field theories are constructed, the due exception being asymptotically AdS spacetimes, which we will not discuss here in details. A reader interested in this topic from the point of view of algebraic quantum field theory can refer to Ref. \citen{Ribeiro:2007hv}. For our goals it is important to stress that a globally hyperbolic spacetime is always strongly causal, hence no pathologies, such as closed causal curves can occur, {\it cf.} Lemma 8.3.8 in \citen{Wald}. Furthermore $\Sigma$ is the natural candidate to play the role of the hypersurface on which to assign the initial data for a partial differential equation ruling the dynamics of the quantum fields, we are interested in. Yet, according to Definition \ref{globhyp}, only the existence of a single Cauchy hypersurface is guaranteed. This is slightly disturbing since there is no reason a priori why an initial value hypersurface for a certain partial differential equation should be distinguished. This quandary has been overcome proving that, if a spacetime $M$, with a smooth metric $g$, is globally hyperbolic, then it is homeomorphic to $\mathbb{R}\times\Sigma$, where $\Sigma$ is a codimension $1$ topological submanifold of $M$ such that, for all $t\in\mathbb{R}$, the locus $\{t\}\times\Sigma$ is a Cauchy hypersurface, {\it cf.} Theorem 3.17 in Ref. \citen{BEE}.

Nonetheless, at this stage, there are at least two potential problems: The first concerns the degree of regularity of $\Sigma$, since we would like to assign smooth initial data, which is not possible if $\Sigma$ is only a continuous hypersurface.  The second is related to the practical use of Definition \ref{globhyp}. It does not suggest any concrete criterion to establish whether a certain spacetime $M$ with an assigned metric $g$ is globally hyperbolic or not. An alternative, yet equivalent, definition of global hyperbolicity requires that $M$ is strongly causal and, for all $p,q\in M$, the set $J_M^+(p)\cap J_M^-(q)$ is either empty or compact, {\it cf.} Definition 3.15 in Ref. \citen{BEE}. Unfortunately, as for Definition \ref{globhyp}, it is rather complicated to directly check these two conditions for a given pair $(M,g)$. 

While the second issue has been recognized as such, the first one was often neglected or, on the basis of flawed proofs, it has been assumed that $\Sigma$ can be made smooth. Only a few years ago, a breakthrough appeared in Refs. 
\citen{Bernal, Bernal:2005qf}. By using mainly deformation arguments, Bernal and Sanchez managed to provide an additional characterization of globally hyperbolic spacetimes. We shall report it following the formulation given in Section 1.3 of Ref. \citen{BGP}:

\begin{theo}\label{BS}
Let $(M,g)$ be any time-oriented spacetime. The following two statements are equivalent:
\begin{enumerate}[(i)]
\item $(M,g)$ is globally hyperbolic;
\item $(M,g)$ is isometric to $\mathbb{R}\times\Sigma$ with line element $ds^2=-\beta dt^2+h_t$, where $t$ runs over the whole $\mathbb{R}$, $\beta\in C^\infty(M)$ is strictly positive, whereas $h_t$ is a smooth Riemannian metric on $\Sigma$ depending smoothly on $t$. Furthermore each $\{t\}\times\Sigma$ is a smooth spacelike Cauchy hypersurface in $M$. 
\end{enumerate}
\end{theo}

Notice that, thanks to this theorem, not only it is clear that it is always possible to choose smooth spacelike Cauchy hypersurfaces, but, we have now at hand a criterion which makes simpler to verify whether a certain $(M,g)$ is globally hyperbolic. For completeness we provide a few examples of globally hyperbolic spacetimes, in order to convince the reader, who is not familiar with this concept, that such class of manifolds contains most of, if not all, the physically interesting examples. 

\begin{ex}
The following spacetimes are globally hyperbolic:
\begin{itemize}
\item All Friedmann-Robertson-Walker spacetimes. These are homogeneous and isotropic solutions of Einstein's equations whose line element is $ds^2=-dt^2+a^2(t)\left[\frac{dr^2}{1-kr^2}+r^2\left(d\theta^2+\sin^2\theta d\varphi^2\right)\right]$. Here $k$ is a constant which can be normalized to $0$, $1$ or $-1$ and, accordingly, the background is $\mathbb{R}\times\Sigma$, where $\Sigma$ is locally homeomorphic to $\mathbb{R}^3$, $\mathbb{S}^3$, the $3$-sphere or $\mathbb{H}^3$, the three-dimensional hyperboloid. The function $a(t)$, known as scale factor, is a smooth and positive scalar function defined on an open interval $I\subseteq\mathbb{R}$. Notice that if $I=(a,b)\subset\mathbb{R}$, one can always redefine $t$ as a new variable $t'$ whose domain of definition is the whole real axis. For example, if $-\infty<a<b<\infty$, then choose $t'=\ln(t-a)-\ln(t-b)$ and the line element has still the form given in Theorem \ref{BS}. Notice that Minkowski spacetime is a special case of this class and, thus, it is, as expected, globally hyperbolic.
\item Spherically symmetric solutions of the vacuum Einstein's equations form a one-parameter family of manifolds, all topologically equivalent to $\mathbb{R}\times I\times\mathbb{S}^2$, where $I=(2M,\infty)$. $M\geq 0$ is the above mentioned parameter, which can be interpreted from a physical point of view as the mass of a body, source of the gravitational field. The line element is $ds^2=-\left(1-\frac{2M}{r}\right)dt^2+\frac{dr^2}{1-\frac{2M}{r}}+r^2\left(d\theta^2+\sin^2\theta d\varphi^2\right)$. Here $(\theta,\varphi)$ are the standard coordinates associated to the $2$-sphere, $t$ plays the role of the time variable and runs over the whole real axis, whereas $r\in (2M,\infty)$. On account of Theorem \ref{BS}, one can prove that, regardless of the value of $M$, Schwarzschild spacetime is static and globally hyperbolic. The limit case $M=0$ coincides with Minkowski spacetime. A similar statement holds true for the static regions both of a Reissner-Nordstr\"om and of a Kerr black hole.
\item The maximally symmetric solution of Einstein's equations with a positive cosmological constant $\Lambda$ goes under the name of de Sitter spacetime $dS_4$. It can be constructed as the set of points $x_\mu\in\mathbb{R}^5$, $\mu=0,...,4$ such that $-x^2_0+\sum_{i=1}^4x^2_i=R^2$ where $R^2=\frac{3}{\Lambda}$. One of the possible clever choices of coordinates\cite{Moschella} shows that $dS_4$ is isometric to $\mathbb{R}\times\mathbb{S}^3$ with line element $ds^2=-dt^2+R^2\cosh\left(\frac{t}{R}\right)\left(d\chi^2+\sin^2\chi\left(d\theta^2+\sin^2\theta d\varphi^2\right)\right)$. Here $t\in\mathbb{R}$ is the time variable whereas $(\chi,\theta,\varphi)$ are the standard coordinates on $\mathbb{S}^3$. One can realize per direct inspection that we have shown that de Sitter spacetime is isometric to a Friedmann-Robertson-Walker spacetime with closed spatial sections. Hence it is globally hyperbolic.
\end{itemize}
\end{ex}

To conclude this section, we introduce some terms which will be often used in the following in order to specify the support properties of the most relevant operators in the study of hyperbolic equations on globally hyperbolic spacetimes.

\begin{defi}
Let $M$ be a globally hyperbolic spacetime and consider a region $\Omega\subseteq M$. We say that $\Omega$ is:
\begin{itemize}
\item {\it spacelike-compact} if there exists a compact subset $K\subseteq M$ such that $\Omega\subseteq J_M(K)$;
\item {\it future- (past-) compact} if its intersection with the causal past (future) of any point is compact, namely if $\Omega\cap J_M^+(p)$ ($\Omega\cap J_M^-(p)$) is compact for each $p\in M$;
\item {\it timelike-compact} if it is both future- and past-compact.
\end{itemize}
\end{defi}

\section{Hyperbolic operators and classical fields}

As anticipated, we are interested in quantum field theories whose underlying dynamics can be described in terms of an initial value problem. The restriction to the class of globally hyperbolic spacetimes ensures the existence both of a family of hypersurfaces on which initial data can be assigned and of a preferred direction of evolution. Yet not all partial differential equations (PDEs) admit an initial value problem which guarantees the existence or the uniqueness of the solution, once suitably regular initial data are assigned. Nor we expect that all possible PDEs can be associated to a physically reasonable system. As we shall see, the structure of the PDE governing the dynamics of a field has to guarantee compatibility with the causal structure of the underlying spacetime and hence we have to make sure that no pathology can incur. The goal of this section is thus to introduce the class of PDEs which are at the core of any quantum field theory on a Lorentzian curved background. This is a topic which has been thoroughly discussed in the literature. We will only sketch the concepts and results we will use in the next sections. Yet we strongly advise a reader interested in further details to refer to Refs. \citen{BGP,Waldmann,Friedlander}. We will follow mostly Ref. \citen{BGP}.

As a starting point and for the sake of completeness, we recall the definitions of a vector bundle and of its sections. The latter play a distinguished role since sections represent the natural mathematical object to associate to the physical idea of a classical field. For a detailed discussion of these topics we refer the reader to the literature, {\it e.g.}~Ref. \citen{Husemoller}, 
Chapter 3, Ref.~\citen{Isham}, Chapter 5 and Ref.~\citen{Jost}, Chapter 1.

\begin{defi}
A {\it vector bundle} of rank $n$ consists of a quadruple $(E,\pi,M,V)$:
\begin{itemize}
\item The {\it base} $M$ is a $d$-dimensional smooth manifold;
\item The {\it typical fiber} $V$ is a $n$-dimensional vector space;
\item The {\it total space} $E$ is a $(d+n)$-dimensional smooth manifold;
\item The {\it projection} $\pi:E\to M$ is a smooth surjective map.
\end{itemize}
These objects must fulfil the following conditions in order to define a vector bundle:
\begin{enumerate}[(i)]
\item Each fiber $E_p=\pi^{-1}(p),p\in M$ is a vector space isomorphic to $V$;
\item For each point $p\in M$ there exists a neighbourhood $U$ of $p$ 
and a diffeomorphism $\Psi:\pi^{-1}(U)\to U\times V$ such that 
$\mathrm{pr}_1\circ\Psi=\pi:\pi^{-1}(U)\to U$, 
$\mathrm{pr}_1:U\times V\to U$ being the projection on the first factor 
of the Cartesian product $U\times V$;
\item $\Psi$ acts as a vector space isomorphism on each fiber, namely 
the map $E_p\to\{p\}\times V, e\mapsto\Psi(e)$ is an isomorphism 
between vector spaces.
\end{enumerate}
Any pair $(U,\Psi)$ fulfilling conditions (ii) and (iii) is called a {\it local trivialization} of 
the vector bundle. Any collection of local trivializations covering the 
base manifold $M$ is called a {\it vector bundle atlas}.
\end{defi}

It is customary to refer to a vector bundle specifying only its total space. 
In the following we adopt this convention whenever this does not lead to 
misunderstandings. Moreover, notice that, although $M$ is an arbitrary smooth manifold, in all applications to quantum field theory we shall choose only those $M$ which are globally hyperbolic spacetimes.

The simplest example of vector bundle and the one implicitly considered in standard quantum field theory on Minkowski spacetime is the Cartesian product $M\times V$ between a manifold $M$ and a vector space $V$. In this case the quadruple 
specifying the vector bundle is $(M\times V,\mathrm{pr}_1,M,V)$ and we have a 
local trivialization $(M,\mathrm{id}_{M\times V})$ provided by the whole manifold 
and by the identity map on the Cartesian product. Since the trivialization is defined on 
the entire base manifold, it is called {\it global} and we say that the 
vector bundle is {\it globally trivial}.

Notice that one can use trivial vector bundles to give a different characterization of 
vector valued maps on a manifold. Suppose we are given a function $f:M\to V$. 
Taking into account the trivial vector bundle $M\times V$, we can define a new function 
$\widetilde{f}:M\to M\times V, p\mapsto(p,f(p))$. We have thus established a bijection $f\mapsto\widetilde{f}$ between vector valued functions and a special class of vector bundle valued functions fulfilling the condition $\mathrm{pr}_1\circ\widetilde{f}=\mathrm{id}_M$.

Another example of vector bundles, widely used in physics is the tangent space $TM$ to a manifold $M$. In this context, $TM$ is in general not trivial. A typical textbook example is the 2-sphere $\mathbb{S}^2$. Yet, it is noteworthy that, since we shall be considering only (four-dimensional) globally hyperbolic spacetimes $M$, these manifolds do have a trivial tangent space, or, in more technical words, they are all parallelizable\cite{Geroch}. Hence $TM$, as a vector bundle, is isomorphic to $M\times\mathbb{R}^4$.

For a non-trivial vector bundle $(E,\pi,M,V)$, we cannot establish a correspondence between 
$V$-valued functions on $M$ and vector bundle valued functions 
on $M$ such that the value of the function at a point $p\in M$ is an element 
in the fiber $E_p$ over $p$.
This fact justifies the introduction of the notion of a section of a vector bundle 
as a generalization of the notion of a vector-valued function on a manifold. This just 
consists in replacing the target space with a vector bundle (instead of a vector space) 
and requiring that the value at a point lies in the fiber over that point. Formally we 
have the following definition.

\begin{defi}\label{defSection}
Let $(E,\pi,M,V)$ be a vector bundle of rank $n$. A section of $(E,\pi,M,V)$ is a 
smooth function $\sigma:M\to E$ such that $\pi\circ\sigma=\mathrm{id}_M$.
We denote the vector space of sections of $(E,\pi,M,V)$ with the symbol 
$\Gamma(M,E)$, while the vector space of compactly supported sections is denoted by 
$\Gamma_0(M,E)$.
\end{defi}

For later purposes, we will focus the attention on sequential continuity of certain maps acting 
on smooth sections of a vector bundle. To this avail, we need to define the notions of 
convergence which are induced on these spaces by their canonical topologies. Notice that 
these are those which are relevant for the theory of distributions.

\begin{defi}\label{defConv}
Let $(E,\pi,M,V)$ be a vector bundle and consider a covariant derivative $\nabla$ acting on 
sections of $E$. We say that a sequence $\{\sigma_n\}\subseteq\Gamma(M,E)$ of smooth 
sections is {\it $\mathcal{E}$-convergent} to $\sigma\in\Gamma(M,E)$ if $\sigma_n$ 
converges uniformly with all its covariant derivatives of arbitrary order to $\sigma$ over any 
compact subset of $M$.

For a sequence $\{\tau_n\}\subseteq\Gamma_0(M,E)$ of smooth 
sections with compact support we say that it is {\it $\mathcal{D}$-convergent} to $\tau\in\Gamma_0(M,E)$ 
if there exists a compact subset $K$ of $M$ including the support of $\tau$ and $\tau_n$ for 
all $n$ and if $\tau_n$ converges with all its covariant derivatives of arbitrary order to 
$\sigma$ uniformly over $K$ (hence everywhere).
\end{defi}

Note that here the choice of a covariant derivative and the choice of a norm 
in the fibers (which is even not explicitly mentioned) are completely irrelevant 
since any choice leads to an equivalent notion of 
convergence (an equivalent topology, in fact). For more details on these aspects see 
Ref.~\citen{Friedlander}, Ref.~\citen{Hormander1}, Chapter 2, and also Ref.~\citen{BGP}, 
Section 1.1 for a short introduction.

One can easily extend algebraic operations defined on vector spaces to vector bundles, just 
performing the algebraic operation fiberwise. For example we have the following constructions:
\begin{itemize}
\item {\it Dualization:} Given a vector bundle $E$, we define its dual $E^\ast$ setting 
the fiber $(E^\ast)_p$ over $p\in M$ to be the dual $(E_p)^\ast$ of the fiber $E_p$;
\item {\it Tensor product:} Given two vector bundles $E$ and $F$ over the same manifold $M$, 
for every point $p\in M$, we take the tensor product of the fibers, $E_p\otimes F_p$. Setting 
$(E\otimes F)_p=E_p\otimes F_p$ defines the vector bundle $E\otimes F$;
\item {\it Direct (or Whitney) sum:} Let $E$ and $F$ be vector bundles over $M$. We define the direct sum 
$E\oplus F$ setting $(E\oplus F)_p=E_p\oplus F_p$ for each $p\in M$;
\item {\it Bundle of homomorphisms:} Let $E$ and $F$ be vector bundles over $M$. To define 
the vector bundle $\mathrm{Hom}(E,F)$ of fiberwise homomorphism from $E$ to $F$, we set 
its fiber over any $p$ to coincide with the space of homomorphisms 
$\mathrm{Hom}(E_p,F_p)$, namely the vector space of linear maps from $E_p$ to $F_p$.
\end{itemize}

As it has been put forward in the literature, {\it e.g.}~Ref. \citen{BG2}, a classical linear field theory on a curved spacetime is completely specified in terms of a vector bundle $E$ endowed with a non-degenerate bilinear form (specifying the kinematical structure) and a differential operator $P$ on $E$ (specifying the dynamics). We start introducing the former concept.

\begin{defi}\label{defInnerProd}
Let $(E,\pi,M,V)$ be a real vector bundle. A {\it Bosonic (Fermionic) non-degenerate bilinear form} on $E$ 
is a smooth real-valued map $\langle\cdot,\cdot\rangle$ on the vector 
bundle $E\otimes E$ such that for each $p\in M$ the following properties holds:
\begin{enumerate}[(i)]
\item On the fiber $E_p\otimes E_p$ over $p$, $\langle\cdot,\cdot\rangle$ is a symmetric (antisymmetric)  bilinear form;
\item If $v\in E_p$ is such that $\langle v,w\rangle=0$ for each $w\in E_p$, then $v=0$.
\end{enumerate}
\end{defi}

Notice that one can interpret condition (i) in the last definition saying that 
the bilinear form $\langle\cdot,\cdot\rangle$ is a section of $E^\ast\otimes E^\ast$. In the following all vector bundles are assumed to be endowed with a non-degenerate bilinear form.

As anticipated, if the base manifold $M$ is orientable, as for globally hyperbolic spacetimes, there exists a non-vanishing volume form on $M$. This allows to define a non-degenerate pairing between smooth sections and compactly supported smooth sections of a vector bundle, provided the latter is endowed with a non-degenerate inner 
product along the fibers:
\begin{equation}\label{eqSectPairing}
(\cdot,\cdot):\Gamma_0(M,E)\otimes\Gamma(M,E)\to\mathbb{R},\quad
(\sigma,\tau)\mapsto\int_M\mathrm{vol}_M\,\langle\sigma,\tau\rangle.
\end{equation}
We will often exploit that the map $(\cdot,\cdot)$ 
can be extended to a subset of $\Gamma(M,E)\otimes\Gamma(M,E)$ larger then 
$\Gamma_0(M,E)\otimes\Gamma(M,E)$, namely the set of linear combinations of 
elements $(\sigma,\tau)\in\Gamma(M,E)\otimes\Gamma(M,E)$ such that 
$\mathrm{supp}(\sigma)\cap\mathrm{supp}(\tau)\subseteq M$ is compact (notice that 
this subset is not a subspace, since the sum is not an internal operation). In fact, one can 
even go further taking into account sections which ``decrease fast enough at infinity''.

\subsection{Differential operators and wave equations}

We shall now focus on defining linear partial differential operators on vector bundles. A thorough discussion of the general theory 
is included in the seminal textbooks of H\"ormander, Refs.~\citen{Hormander1, Hormander2, Hormander3, Hormander4}. In our case, instead, we will be mostly interested in a very special subclass of operators, which are those associated to the dynamics of free field theories on curved backgrounds. The prototype of an element in such subclass is the wave operator, to which analysis several complete reviews have been dedicated\cite{BG1, BGP, Waldmann}.

\begin{defi}\label{defLinPartDiffOp}
Let $(E,\pi,M,V),(F,\rho,M,W)$ be two vector bundles over the same $d$-dimensional 
manifold $M$. A {\it linear partial differential operator of order at most $k$} is a linear 
operator $L:\Gamma(M,E)\to\Gamma(M,F)$ fulfilling the following property: 
For each $x\in M$, there exists a neighbourhood $U$ of $x$ such that $(U,\Phi)$ trivializes 
$E$, $(U,\Psi)$ trivializes $F$, $(U,\phi)$ is a local chart of $M$ and there is a collection 
$\{A,A_{j_1},\dots,A_{j_1,\dots,j_k}\;|\;j_1,\dots,j_k\in\{1,\dots,d\}\}$ 
of smooth $\mathrm{Hom}(V,W)$-valued maps on $\phi(U)$
which allows to express $L$ locally: For each section $\sigma$ of $E$
\begin{equation}
\Psi\circ(L\sigma)\circ\phi^{-1}=\sum_{i=0}^k \sum_{j_1,\dots,j_i=1}^d A_{j_1,\dots,j_i}
\partial_{j_1}\cdots\partial_{j_i}(\Phi\circ\sigma\circ\phi^{-1}),
\end{equation}
where $\partial_r,r\in\{1,\dots,d\}$, is the standard partial derivative acting on 
vector valued functions defined on some open subset of $\mathbb{R}^d$.

We say that $L$ is {\it exactly of order $k$} if it is of order at most $k$, but not of order 
at most $k-1$.
\end{defi}

Notice that, as a consequence of their definition, linear partial differential operators cannot 
enlarge the support of a section. This property will be often implicitly used in the following.

We postpone examples of linear partial differential operators to the next subsection and we focus our attention on a special subclass. This class is noteworthy since, on the one hand, it allows a generalization of the standard picture of a wave equation we have on Minkowski spacetime (a prominent example is the Klein-Gordon equation), while, on the other hand, it admits a well-behaved initial value problem, namely existence and uniqueness of a solution hold once initial data are properly assigned on a Cauchy hypersurface. Nonetheless, it is important to bear in mind that the Dirac operator is not included in the class we are considering, hence an extension of our analysis will be necessary.

\begin{defi}\label{nhyp}
Let $(E,\pi,M,V)$ be a real vector bundle over a $d$-dimensional Lorentzian manifold $(M,g)$. 
We say that a partial differential operator $P:\Gamma(M,E)\to\Gamma(M,E)$ of second order 
is {\it normally hyperbolic} if, in a local trivialization,\footnote{Refer to Definition \ref{defLinPartDiffOp} 
for a choice of the neighbourhood $U$ and of the maps $\phi$ and $\Phi$.} 
there exists a collection $\{A,A_i\;|\;i\in\{1,\dots,d\}\}$ of smooth $\mathrm{Hom}(V,V)$-valued maps on $U$ 
such that $P$ reads as follows: 
For each section $\sigma$ of $E$,
\begin{equation}
\Phi\circ(P\sigma)\circ\phi^{-1}=\left(-\sum_{i,j=1}^d g^{ij}\mathrm{id}_V
\partial_i\partial_j+\sum_{i=1}^d A_i\partial_i+A\right)(\Phi\circ\sigma\circ\phi^{-1}).
\end{equation}

Given a section $J$ of a vector bundle $E$, called {\it source}, and a differential operator $P$ taking values on $\Gamma(M,E)$, 
we say that a partial differential equation $P\sigma=J$ is a {\it wave equation} if $P$ is normally hyperbolic.
\end{defi}

Notice that this definition enforces that a normally hyperbolic operator has order two 
and that its coefficients of highest order should coincide with the metric of the underlying 
Lorentzian manifold.

Definition \ref{nhyp} is usually formulated introducing the {\it principal symbol} associated to a partial differential operator. Although such approach is more elegant and more apt to generalizations, we will not pursue it here. It would add an additional degree of mathematical complexity which is not strictly necessary. An interested reader can refer to Ref. \citen{Hormander1}.

\begin{rem}
Notice that, wave equations, for which $J\neq 0$, are of inhomogeneous type and their space of solutions is an affine space. In this case the most natural mathematical structure is provided by affine bundles. Quantum field theories on affine bundle have been developed in Ref.~\citen{Benini:2012vi} and they provide also the natural setting for the quantization of Abelian Yang-Mills models seen as theories for the connections of suitable principal bundles\cite{Benini:2013tra, Sanders:2012sf}. In this review, we will only focus on the case without sources, $J=0$, for which the set of all solutions has indeed a vector space structure.
\end{rem}

The simplest example of a wave equation one can consider is the real scalar field with mass 
$m\geq0$. Let $(M,g)$ be a $d$-dimensional Lorentzian manifold and consider the trivial line 
bundle $M\times\mathbb{R}$. As shown previously, sections of this bundle can be identified 
with real valued functions on $M$. We define the normally hyperbolic operator 
$P:C^\infty(M)\to C^\infty(M)$ in a local chart according to the following formula:
\begin{equation}\label{eqScalarFieldOp}
P=-\sum_{i,j=1}^d g^{ij}\partial_i\partial_j+\sum_{i=1}^d (\sum_{j,k=1}^d
g^{jk}\Gamma_{jk}^i)\partial_i+m^2,
\end{equation}
where $\Gamma_{ij}^k$ denotes the Christoffel's symbols of the Levi-Civita connection on $(M,g)$ 
(refer to Ref.~\citen{Jost}, Section 3.3, 
for an introduction to the Levi-Civita connection and its Christoffel's symbols).
Notice that $Pf=0$ reduces to the usual wave equation for the real scalar field on Minkowski 
spacetime, where this formula is actually global and Christoffel's symbols become trivial.

In the next subsection we will need another construction, essentially based on Stokes' theorem 
(integration by parts), the {\it formal adjoint} of a differential operator. This involves 
the pairing between sections introduced in Eq. \eqref{eqSectPairing} exploiting the bilinear form 
defined on the fibers of the underlying vector bundle (see Definition \ref{defInnerProd}).

\begin{defi}\label{for_adj}
Given a linear differential operator $L:\Gamma(M,E)\to\Gamma(M,F)$, its {\it formal adjoint} 
$L^\ast:\Gamma(M,F)\to\Gamma(M,E)$ is a linear differential operator satisfying 
$(L^\ast u,v)_E=(u,Lv)_F$ for each $u\in\Gamma(M,F)$ and $v\in\Gamma(M,E)$ with 
$\mathrm{supp}(u)\cap\mathrm{supp}(v)$ compact, where $(\cdot,\cdot)_E$ 
and $(\cdot,\cdot)_F$ are the non-degenerate pairings introduced in 
Eq. \eqref{eqSectPairing} for $E$ and respectively $F$.

A linear differential operator $L:\Gamma(M,E)\to\Gamma(M,E)$ is {\it formally self-adjoint} 
if it coincides with its formal adjoint, namely $L^\ast=L$.
\end{defi}

The existence of $L^\ast$ follows from Stokes' theorem, while 
its uniqueness is a direct consequence of the non-degeneracy of the pairing 
$(\cdot,\cdot)$ between sections and compactly 
supported sections of a vector bundle. Indeed the operation of taking 
the formal adjoint is an involution, namely $(P^\ast)^\ast=P$.

From a physical point of view, formally self-adjoint differential 
operators play a distinguished role, as we will see at the end of this section. This property is quite natural as equations of motions, which originate from a Lagrangian, quadratic in the fields, are always written in terms of a formally self-adjoint differential operator.

\begin{rem}\label{remAdjNormHypOp}
The formal adjoint of a normally hyperbolic operator is still normally hyperbolic. 
This follows from the local structure of such class of operators via integration by parts.
\end{rem}

It is an important result, especially in view of the forthcoming theorem on the existence and 
uniqueness of solutions of initial value problems defined using a normally hyperbolic operator, 
that one can always find a connection on the underlying vector bundle which makes it possible to express any 
given normally hyperbolic operator in a simple form. For an introductory discussion on the notion of 
a connection on a vector bundle, we refer to Ref. \citen{Jost}, Chapter 3.

\begin{propo}\label{propBoxDecomposition}
Consider a normally hyperbolic operator $P$ acting on sections of a vector bundle $E$ over a Lorentzian manifold $(M,g)$. 
Then there exist a unique connection $\nabla$ on $E$ 
and a unique section $B$ of the vector bundle $\mathrm{Hom}(E,E)$ such that $P=\Box_\nabla+B$, 
where $\Box_\nabla$ is the d'Alembert operator defined by 
$\Box_\nabla=-(\mathrm{tr}_{T^\ast M}\otimes\mathrm{id}_E)\circ\nabla\circ\nabla$, 
where $\mathrm{tr}_{T^\ast M}$ denotes the trace on $T^\ast M\otimes T^\ast M$ with respect to the metric $g$ 
and $\mathrm{id}_E$ is just the identity on $E$.
\end{propo}

\noindent For the proof of the last proposition see Ref. \citen{BGP}, Lemma 1.5.5.

This proposition clarifies the fact that d'Alembert operators associated to connections are 
the prototypes of all possible normally hyperbolic operators up to a term of order zero in 
the derivatives. This is exactly what happens in the only example we explicitly considered so far, 
namely the real scalar field of Eq. (\ref{eqScalarFieldOp}). In fact, taking into account the Levi-Civita 
connection $\nabla$ for the metric $g$ on the base manifold $M$, locally
$P=-g^{ij}\nabla_i\nabla_j+m^2$,
which is exactly the expression for $\Box_{\nabla}+m^2$ in a local chart of $M$.

Wave equations on globally hyperbolic spacetimes satisfy 
a very special property, namely each Cauchy problem admits a unique 
global solution. We make this statement precise in the following theorem. 
For the proof refer to Refs.~\citen{BGP, Friedlander}.

\begin{theo}\label{thmCauchy}
Let $(M,g)$ be a globally hyperbolic spacetime and consider a spacelike smooth Cauchy hypersurface $\Sigma$ 
of $(M,g)$. Let $\mathfrak{n}$ be the future-directed timelike unit normal vector field on $\Sigma$. 
Consider a vector bundle $(E,\pi,M,V)$ and a normally hyperbolic operator 
$P=\Box_\nabla+B:\Gamma(M,E)\to\Gamma(M,E)$, where $\nabla$ and $B$ are uniquely determined 
according to Proposition \ref{propBoxDecomposition}. 
For each initial data $u_0,u_1\in\Gamma_0(\Sigma,E)$ and each source $f\in\Gamma_0(M,E)$, 
the Cauchy problem
\begin{equation}
Pu=f\mbox{ on }M,\quad\nabla_\mathfrak{n}u=u_1\mbox{ on }\Sigma,\quad u=u_0\mbox{ on }\Sigma,
\end{equation}
admits a unique solution $u\in\Gamma(M,E)$. 
Here $\nabla_\mathfrak{n}$ denotes the covariant derivative along $\mathfrak{n}$.

The support of $u$ is related to the supports of both initial data and source, namely
\begin{equation}
\mathrm{supp}(u)\subseteq J_M(\mathrm{supp}(u_0)\cup\mathrm{supp}(u_1)\cup\mathrm{supp}(f)).
\end{equation}

Moreover the map $\Gamma_0(\Sigma,E)\times\Gamma_0(\Sigma,E)\times\Gamma_0(M,E)\to
\Gamma(M,E),(u_0,u_1,f)\mapsto u$, which assigns the unique solution of the corresponding 
Cauchy problem to a given set of initial data and a given source, is linear and continuous.
\end{theo}

The last statement in Theorem \ref{thmCauchy} makes precise the idea that solutions of a Cauchy 
problem depend continuously on the given initial data on a Cauchy surface. 
As it will be clear soon, this fact entails continuity of the Green operators, which are one of the 
essential ingredients for quantization. These operators are the main topic of the next subsection.

\subsection{Green operators}\label{subsectGreenOp}
Theorem \ref{thmCauchy} is an extremely relevant result. It implies in particular the existence and uniqueness of the so-called advanced and retarded Green operators. As we will see, we can use Green operators to completely characterize the space of solutions with spacelike compact support of a wave equation on a globally hyperbolic spacetime.

\begin{defi}
Let $E$ be a vector bundle over a time-oriented Lorentzian manifold $(M,g)$ and consider a linear 
differential operator $P:\Gamma(M,E)\to\Gamma(M,E)$. A linear map $G^\pm:\Gamma_0(M,E)\to
\Gamma(M,E)$ is an {\it advanced/retarded Green operator} for $P$ if the following conditions are 
satisfied for each $f\in\Gamma_0(M,E)$:
\begin{enumerate}[(i)]
\item $PG^\pm f=f$;
\item $G^\pm Pf=f$;
\item $\mathrm{supp}(G^\pm f)\subseteq J_M^\pm(\mathrm{supp}(f))$.
\end{enumerate}
\end{defi}

Not all linear differential operators admit Green operators, but, indeed, those who do, are of primary physical relevance, since, as we shall show in the next section, one can associate to them a distinguished algebra of observables, built out of $G^\pm$. For this reason we can encompass these special operators in a specific class:

\begin{defi}\label{Ghyp}
Let $E$ be a vector bundle over a time-oriented Lorentzian manifold $(M,g)$ and consider a linear 
differential operator $P:\Gamma(M,E)\to\Gamma(M,E)$. $P$ is called {\it Green-hyperbolic} 
if it admits advanced and retarded Green operators.
\end{defi}

In this definition, we do not require uniqueness of the Green operators. It holds true automatically, when both $P$ and its formal adjoint $P^\ast$ are Green-hyperbolic. This is due to the properties of advanced and retarded Green operators as well as to the non degeneracy of the pairing between sections.\footnote{We are grateful to Ko Sanders for pointing out to us that the argument used in Ref. \citen{BG2} to prove uniqueness implicitly relies on the fact that both $P$ and its adjoint must be Green-hyperbolic.}

Normally hyperbolic operators are distinguished from the physical viewpoint because,
as a consequence of Theorem \ref{thmCauchy}, they provide a very large class of Green-hyperbolic operators: 
Consider a vector bundle $E$ on a globally hyperbolic spacetime $(M,g)$ and assume that a 
normally hyperbolic operator $P$ acting on sections of $E$ is given. For any compactly supported 
section $f$, one fixes a spacelike Cauchy surface $\Sigma_-$ of $(M,g)$ which is disjoint from 
the support of $f$, but such that its causal future includes it. One can thus set up a Cauchy problem with vanishing 
initial data on $\Sigma_-$ and with source $f$. From Theorem \ref{thmCauchy} we know that there exists 
a unique solution $u_f$ with support included in $J_M(\mathrm{supp}(f))$. One can prove, moreover, 
that $\mathrm{supp}(u_f)$ is included in the causal future of the support of $f$: Consider the same 
Cauchy problem restricted to the globally hyperbolic subspacetime 
$M_-=M\setminus J_M^+(\mathrm{supp}(f))\subseteq M$. Indeed, a solution for this problem is given by the 
restriction of $u_f$ to the subspacetime $M_-$. But this is also a Cauchy problem with vanishing initial data 
and source, namely the null section is a solution. By uniqueness it follows that $u_f=0$ 
on $M_-$, which entails that $\mathrm{supp}(u_f)\subseteq J_M^+(\mathrm{supp}(f)$. 
Using a similar argument, it follows that this construction is independent of the choice of $\Sigma_-$. 
Hence the construction defines a map $\Gamma_0(M,E)\to\Gamma(M,E),f\mapsto u_f$ 
such that $Pu_f=f$ and $\mathrm{supp}(u_f)\subseteq J_M^+(\mathrm{supp}(f))$ for each $f\in\Gamma_0(M,E)$.

Theorem \ref{thmCauchy} entails linearity of this map. 
Per construction this map fulfills the first and the third properties of an advanced 
Green operator for $P$. The second property follows since, for any spacelike Cauchy surface 
$\Sigma$ disjoint from the causal future of the support of $Pf$, $u_{Pf}$ is a solution of the Cauchy problem 
$Pu_{Pf}=Pf$ with vanishing initial data on $\Sigma$. Consequently $u_{Pf}-f$ is a solution of a 
Cauchy problem with vanishing initial data on $\Sigma$ and vanishing source, meaning that $u_{Pf}-f=0$. Hence the 
map $f\mapsto u_f$ is an advanced Green operator for $P$. We denote this map with $G^+$. 
By exchanging past and future in the argument above one obtains a retarded Green operator $G^-$ for $P$. 
Hence we proved that all normally hyperbolic operators are Green-hyperbolic.

It turns out that uniqueness of the advanced and retarded Green operator for a normally hyperbolic 
operator automatically holds true since, according to Remark \ref{remAdjNormHypOp}, 
$P^\ast$ is normally hyperbolic too, Green-hyperbolic in particular. 
Our discussion about uniqueness for the more general case of a Green-hyperbolic operator, whose 
adjoint is Green-hyperbolic as well, applies and one can easily show that that $G^{\ast\pm}$ 
is formally adjoint\footnote{Here the condition of formal adjointness is intended only with respect to 
compactly supported sections.} to $G^\mp$.

We recollect the results of the discussion above in the following theorem.

\begin{theo}\label{thmGreenOp}
Let $E$ be a vector bundle over a globally hyperbolic spacetime $(M,g)$ and consider a 
Green-hyperbolic operator $P$ acting on sections of $E$. Moreover, assume that its formal adjoint 
$P^\ast$ is Green-hyperbolic too. Then advanced and retarded Green operators for both $P$ and 
$P^\ast$ exist and are unique. Denote with $G^\pm$ those for $P$ and with $G^{\ast\pm}$ 
those for $P^\ast$. We have that $(G^{\ast\pm} u,v)=(u,G^\mp v)$ 
for all section $u,v$ of $E$ with compact support.

In particular the same conclusions apply to normally hyperbolic operators 
since they are Green-hyperbolic and their formal adjoint is normally hyperbolic.
\end{theo}

We can use $G^+$ and $G^-$ in order to introduce a new operator $G=G^+-G^-$. 

\begin{defi}
Let $E$ be a vector bundle over a time-oriented Lorentzian manifold $(M,g)$. 
Consider a Green-hyperbolic operator $P$ acting on sections of $E$ and
take advanced and retarded Green operators $G^\pm$ for 
$P$.\footnote{Notice that without further assumptions on $P$, 
there might be more then one choice for $G^\pm$.}
Then $G=G^+ - G^-$ is the {\it causal propagator} for $P$ defined by $G^\pm$.
\end{defi}

$G$ entails the full characterization of the space of solutions with 
spacelike-compact support of the equation 
$Pu=0$ for a normally hyperbolic operator acting on sections of a vector bundle $E$ over a 
globally hyperbolic spacetime $(M,g)$. In fact, the forthcoming argument 
will be valid for Green-hyperbolic operators with Green-hyperbolic formal adjoint 
(notice that, under this assumption, the Green operators are unique as well as 
the causal propagator). There exists also a rather remarkable extension of this result 
to the full space of solutions of $Pu=0$ without any further assumption on the support 
of solutions and also admitting distributional solutions. Such an extension relies heavily on the support properties of advanced and retarded Green operators, see Ref. \citen{Sanders:2012}, Section 5.

Notice the following properties of $G$, which descend from those of $G^\pm$. For each $f\in\Gamma_0(M,E)$, the 
following holds:
\begin{enumerate}[(i)]
\item $PGf=0$;
\item $GPf=0$;
\item $\mathrm{supp}(Gf)\subseteq J_M(\mathrm{supp}(f))$.
\end{enumerate}

On account of the third property, the image of $G:\Gamma_0(M,E)\to\Gamma(M,E)$ is 
contained in $\Gamma_{sc}(M,E)$, namely the {\it space of sections with spacelike-compact support}, 
those sections whose support is included in $J_M(K)$ for some compact subset $K$ of $M$. Hence 
one can consider a new map $\Gamma_0(M,E)\to\Gamma_{sc}(M,E),f\mapsto Gf$, still denoted by $G$ 
with a slight abuse of notation, which coincides with the old one up to the inclusion of $\Gamma_{sc}(M,E)$ into $\Gamma(M,E)$.

The first property entails that the space of solutions of $Pu=0$, which, by definition, 
coincides with the kernel of $P$, includes the image of the operator $G$. Moreover, 
if $u$ is such that $Pu=0$ and its support is spacelike-compact, we can consider 
$u_+=\chi_+u$ and $u_-=\chi_-u$, where $\{\chi_+,\chi_-\}$ is a partition of unity subordinate to the 
open cover $\{I_M^+(\Sigma_-),I_M^-(\Sigma_+)\}$ and $\Sigma_+,\Sigma_-$ are disjoint spacelike 
Cauchy surfaces, with $\Sigma_+$ lying in the future of $\Sigma_-$. Per linearity $Pu_++Pu_-=Pu=0$. 
Together with the support properties of $u$ (spacelike-compact support) 
and $\chi_+,\chi_-$ (past-compact, respectively future-compact, support), this identity entails that $Pu_+=-Pu_-$ 
has compact support. Hence we can consider $G^\pm Pu_+=-G^\pm Pu_-$\footnote{Notice that $u_+$ 
and $u_-$ are not supposed to be compactly supported, so one cannot conclude, for example, that 
$G^+ P u_\pm=u_\pm$. In fact, if both $u_+$ and $u_-$ were compactly supported, 
this would immediately lead to $u=0$. This is consistent with the fact that $u$ would have 
compact support too, hence the assumption $Pu=0$ entails $u=G^+Pu=0$.} 
and define $v=GPu_+$. On account of the relation between the Green operators 
for $P$ and $P^\ast$, we deduce the following chain of identities  
for each section $f\in\Gamma_0(M,E)$:
\begin{equation}
(f,v)=(f,G^+ P u_+ + G^- P u_-)=(P^\ast G^{\ast-} f,u_+) + (P^\ast G^{\ast+} f,u_-)=(f,u).
\end{equation}
Non-degeneracy of $(\cdot,\cdot)$ entails $v=u$. This proves that $u$ lies in 
the image of $G$, showing that the space of solutions 
with spacelike-compact support of the equation $Pu=0$ coincides with the 
image of the causal propagator $G$.

Moreover the second property states that the image of the normally hyperbolic operator $P$ acting on 
sections with compact support, namely $P(\Gamma_0(M,E))\subseteq\Gamma_0(M,E)$, is contained 
in the kernel of the causal propagator $G$. At the same time, assuming that $f\in\Gamma_0(M,E)$ is such 
that $Gf=0$, we deduce that $G^+f=G^-f$ is a section whose support is included in both 
the causal past and the causal future of the support of $f$. 
Hence $u=G^+f$ is compactly supported and satisfies $Pu=f$, showing that 
$P(\Gamma_0(M,E))$ exactly coincides with the kernel of $G$.

We recollect all these facts in the following statement.

\begin{propo}\label{propCausalProp}
Consider a vector bundle $E$ over a globally hyperbolic spacetime $(M,g)$. Let 
$P:\Gamma(M,E)\to\Gamma(M,E)$ be a Green-hyperbolic operator with 
Green-hyperbolic formal adjoint $P^\ast$ and denote with $G^+,G^-$ 
its unique Green operators. Then the associated causal propagator 
$G:\Gamma_0(M,E)\to\Gamma_{sc}(M,E)$ has the following properties:
\begin{enumerate}[(i)]
\item The kernel of $G$ coincides with the image of $P$ acting on $\Gamma_0(M,E)$;
\item The image of $G$ coincides with the kernel of $P$ acting on $\Gamma_{sc}(M,E)$.
\end{enumerate}
Moreover $P$ is injective when acting on compactly supported sections.
\end{propo}

\begin{proof}
Everything has already been proven in the discussion above, except the last statement. 
For the last statement, assume that $f\in\Gamma_0(M,E)$ is such that $Pf=0$. 
Then $f=G^+ Pf=0$, which proves the assertion.
\end{proof}

One can interpret Proposition \ref{propCausalProp} as a characterization of 
the space $\mathcal{S}_{sc}(M)$ of solutions with spacelike-compact support of the equation $Pu=0$ on $M$, 
$P$ being Green-hyperbolic with Green-hyperbolic formal adjoint. 
First one has to realize that this space coincides with the kernel of $P$ restricted to 
spacelike-like compact sections, which is a rather trivial fact. Then one can apply Proposition \ref{propCausalProp} 
to deduce that $G$ induces an isomorphism of vector spaces from 
$\Gamma_0(M,E)/P(\Gamma_0(M,E))$ to $\mathcal{S}_{sc}(M)$.

\begin{lem}\label{iso}
Consider a vector bundle $E$ over a globally hyperbolic spacetime $(M,g)$. Let 
$P:\Gamma(M,E)\to\Gamma(M,E)$ be a Green-hyperbolic operator with Green-hyperbolic formal adjoint. 
Then the space $\mathcal{S}_{sc}(M)$ of solutions with spacelike-compact support 
of the equation $Pu=0$ on $M$ is isomorphic to the quotient space $\Gamma_0(M,E)/P(\Gamma_0(M,E))$ 
via the map $I:\Gamma_0(M,E)/P(\Gamma_0(M,E))\to\mathcal{S}_{sc}(M), [f]\mapsto Gf$.
\end{lem}

Besides its mathematical characterization, also the physical interpretation of $\mathcal{S}_{sc}(M)$ is noteworthy. As a matter of fact it plays the role of the {\it classical phase space} of the theory. To support this interpretation we need to show that we can endow $\mathcal{S}_{sc}(M)$ with a symplectic structure (Bosonic case) or an inner product structure (Fermionic case). This is possible for all formally self-adjoint Green-hyperbolic operators.

\begin{propo}\label{sympl2}
Let $E$ be a vector bundle over a globally hyperbolic spacetime endowed with a Bosonic or Fermionic non-degenerate bilinear form. Let $P:\Gamma(M,E)\to\Gamma(M,E)$ be a formally self-adjoint Green-hyperbolic operator. Then the causal propagator $G$ for $P$ fulfils $(u,Gv)=-(Gu,v)$ for each $u,v\in\Gamma_0(M,E)$ and the map $\sigma:\mathcal{S}_{sc}(M)\otimes\mathcal{S}_{sc}(M)\to\mathbb{R}$ defined by $\sigma(u,v)=(f,Gh)$, $f,h\in\Gamma_0(M,E)$ such that $Gf=u,Gh=v$, is a symplectic form, {\it i.e.}~a non-degenerate skew-symmetric bilinear form, in the Bosonic case or an inner product, {\it i.e.}~a non-degenerate symmetric bilinear form, in the Fermionic case.
\end{propo}

\begin{proof}
Since $P$ is Green-hyperbolic and formally self-adjoint, its Green operators $G^+,G^-$ are unique 
and coincide with the Green operators of its formal adjoint $P^\ast=P$, hence we have 
$(G^\pm u,v)=(u,G^\mp v)$ for each $u,v\in\Gamma_0(M,E)$. This follows from 
Theorem \ref{thmGreenOp}. Defining the causal propagator as $G=G^+-G^-$, 
we deduce $(u,Gv)=-(Gu,v)$ for each $u,v\in\Gamma_0(M,E)$.

As for the symplectic form (respectively inner product), we start defining an ancillary bilinear form 
$\tau$ on $\Gamma_0(M,E)$ by setting $\tau(u,v)=(u,Gv)$ for each $u,v\in\Gamma_0(M,E)$. 
We show that the space of degeneracy of $\tau$ is $P(\Gamma_0(M,E))$: Consider $u\in\Gamma_0(M,E)$
such that $\tau(u,v)=0$ for each $v\in\Gamma_0(M,E)$. Non-degeneracy of $(\cdot,\cdot)$ entails $Gu=0$. 
Then $u\in P(\Gamma_0(M,E))$, due to (ii) in Proposition \ref{propCausalProp}. 
The converse follows from $GP=0$ on sections with compact support, 
namely $\tau(Pf,v)=0$ for each $f,v\in\Gamma_0(M,E)$. This means that $\tau$ induces a non-degenerate 
skew-symmetric (Bosonic case), respectively symmetric (Fermionic case), bilinear form on the quotient $\Gamma_0(M,E)/P(\Gamma_0(M,E))$, that is to say, 
a symplectic form, respectively an inner product. Using the isomorphism defined in Lemma \ref{iso}, 
we define $\sigma=\tau\circ(I^{-1}\otimes I^{-1})$. It is straightforward to check 
that $\sigma$ as defined here coincides with the map in the statement.
\end{proof}

To conclude the section, we remark that up to this point we did not provide any information about continuity properties of the operators we introduced so far. However, recalling Theorem \ref{thmCauchy}, one should expect continuity both
of Green operators and of the causal propagator to hold, at least whenever Theorem 
\ref{thmCauchy} can be applied. As a matter of fact, this is the case.

Let us first introduce three notions of convergence which are related to the causal structure 
of spacetime. In order to do this we need two new spaces of sections which refine the notion 
of spacelike-compact sections. Those are $\Gamma_{psc}(M,E)$ and $\Gamma_{fsc}(M,E)$: 
The first one is just the {\it space of sections with past- and spacelike- compact support}, 
namely those sections whose support is contained in the causal future $J^+_M(K)$ of a compact 
subset $K\subseteq M$. $\Gamma_{fsc}(M,E)$ is obtained exchanging future and past in the 
definition of $\Gamma_{psc}(M,E)$. These are indeed both subspaces of $\Gamma_{sc}(M,E)$. 
Notice that, exactly as we are allowed to restrict the codomain of a causal propagator to 
$\Gamma_{sc}(M,E)$, instead of $\Gamma(M,E)$, we can also restrict the codomain of 
an advanced (retarded) Green operator to $\Gamma_{psc}(M,E)$ 
(respectively $\Gamma_{fsc}(M,E)$). The relevant spaces for the analysis of 
the continuity properties of advanced and retarded Green operators, 
as well as causal propagators, being available, we are ready to define convergent sequences on these 
spaces. Note that this notion of convergence is just a refinement of the convergence for sequences in $\Gamma(M,E)$ (see Definition \ref{defConv}). It is obtained imposing constraints on the supports of the sections forming the sequence.

\begin{defi}\label{defConv2}
We say that a sequence $\{f_n\}$ of elements in $\Gamma_{sc}(M,E)$ ($\Gamma_{psc}(M,E)$, 
$\Gamma_{fsc}(M,E)$) is {\it $\mathcal{E}_{sc}$-convergent} (respectively $\mathcal{E}_{psc}$-, 
$\mathcal{E}_{fsc}$-) to $f$ in $\Gamma_{sc}(M,E)$ (respectively 
$\Gamma_{psc}(M,E)$, $\Gamma_{fsc}(M,E)$) if the following holds:
\begin{enumerate}[(i)]
\item There exists a compact subset $K\subseteq M$ such that $\mathrm{supp}(f_n)$, 
for each $n\in\mathbb{N}$, and $\mathrm{supp}(f)$ are included in $J_M(K)$ 
(respectively $J^+_M(K)$, $J^-_M(K)$).
\item $\{f_n\}$ is $\mathcal{E}$-convergent to $f$, namely it converges uniformly to $f$ 
together with all its derivatives on every compact subset of $M$.
\end{enumerate}
\end{defi}

Since we are now assuming $P$ to be normally hyperbolic, the following holds true:
$G^\pm:\Gamma_0(M,E)\to\Gamma_{(p/f)sc}(M,E)$ and 
$G:\Gamma_0(M,E)\to\Gamma_{sc}(M,E)$ are sequentially continuous with respect to the 
$\mathcal{D}$-convergence on $\Gamma_0(M,E)$ (see Definition \ref{defConv}) and the 
convergences defined above.

This fact follows straightforwardly from Theorem \ref{thmCauchy} (we are now restricting 
ourselves only to normally hyperbolic operators because otherwise this theorem may no longer be valid). 
For example, consider $G^+$ and 
take a sequence $\{f_n\}$ in $\Gamma_0(M,E)$ which is $\mathcal{D}$-convergent to $f\in\Gamma_0(M,E)$. 
In particular there exists a compact subset $K\subseteq M$ including 
the supports of the limit and of all elements in the sequence. Take a Cauchy surface $\Sigma$ 
of $M$ such that $K\cap J_M^-(\Sigma)=\emptyset$ (existence of $\Sigma$ follows from 
the fact that $K$ is compact). 
Our assumptions entail that the sequence of Cauchy data 
$\{(0,0,f_n)\}\in\Gamma_0(\Sigma,E)\times\Gamma_0(\Sigma,E)\times\Gamma_0(M,E)$ 
is $\mathcal{D}$-convergent to $(0,0,f)$. For each $n$, we denote with $u_n\in\Gamma(M,E)$ the solution 
of the Cauchy problem defined by $P$ and $\Sigma$ with initial data $(0,0,f_n)$. 
Similarly we denote with $u$ the solution of the Cauchy problem corresponding to $(0,0,f)$. 
Applying Theorem \ref{thmCauchy}, we deduce that 
$\{u_n\}$ is $\mathcal{E}$-convergent to $u$ and $\mathrm{supp}(u)$ and 
$\mathrm{supp}(u_n)$, for all $n$, are included in $J_M^+(K)$. Moreover, notice that 
$u_n=G^+f_n$, for all $n$, and $u=G^+f$, as a consequence of the construction of $G^+$ 
for $P$ normally hyperbolic (see the discussion preceding Theorem \ref{thmGreenOp}). 
This means that $\{u_n\}$ is $\mathcal{E}_{psc}$-convergent to $u$. 
Hence $G^+:\Gamma_0(M,E)\to\Gamma_{psc}(M,E)$ is sequentially continuous. The same argument 
applies to $G^-$ and the result for the causal propagator $G$ follows from sequential 
continuity of the inclusion maps $\Gamma_{(p/f)sc}(M,E)\to\Gamma_{sc}(M,E)$.

If one wants to draw the same conclusions for advanced and retarded Green-operators 
as well as for a causal propagator of a Green-hyperbolic operator $P$, 
one could refine its definition in order to enforce that $P$ must admit 
{\it sequentially continuous} advanced and retarded Green operators; 
as of now we are not aware of a proof that sequential continuity follows from Green-hyperbolicity. 
Yet, let us stress that sequential continuity of the advanced and retarded Green operators 
holds for the two examples of Green-hyperbolic (but not normally hyperbolic) operators 
we will analyse in the following, namely the Majorana and the Proca fields. 
This fact is a direct consequence of their construction.

\subsection{Examples}
We focus our attention to the construction of explicit examples of field theories and particularly to the analysis of their dynamics. 
As we outlined in the introduction, we shall discuss in detail only free field theories and refer the interested reader to Refs. \citen{Brunetti:1999jn,Fredenhagen:2011mq,Fredenhagen:2012sb,Hollands:2007zg} for a discussion of perturbative interacting quantum field theories on curved spacetimes in the algebraic language. At this stage we can see the first great difference between the construction of quantum field theories on a generic curved background and on Minkowski spacetime. In the latter, since the Poincar\'e group encodes all possible isometries, the collection of all possible free field theories, together with their equations of motion, can be fully classified by means of a group theoretical analysis, see for example Ref.~\citen{BR}. No additional input, save for the requirement of the fields to behave covariantly under the action of the Poincar\'e group, is needed. The so-called Bargmann-Wigner construction has no counterpart when the underlying background is a generic globally hyperbolic spacetime since the isometry group can be in general even trivial. For this reason, although it is natural to try to define the same kind of fields, {\it e.g.}~scalar, Dirac, Proca, we have a bigger leeway in deciding which are the relevant equations of motion. The only requirement we can ask for is compatibility with the known ones in the limit when the metric tends to that of Minkowski spacetime. We shall divide our analysis in three cases: scalar, Majorana and Proca fields.

\subsubsection{Scalar fields} 

In order to show the effect of this additional freedom, we start from the simplest example, the {\it real scalar field}. The other examples we shall treat, the Majorana field and the Proca field, are Hermitian and real fields respectively. Complex fields such as Dirac fields can be treated with only a few additional steps. Scalar fields have been at the heart of the vast majority of papers devoted to quantum field theory on curved backgrounds in the algebraic approach, since the first analysis of Dimock\cite{Dimock}. Another very detailed investigation can be found in Ref. \citen{Wald2}. 

As a starting point, we remark that the underlying bundle structure is rather simple. As a matter of fact, given an arbitrary globally hyperbolic spacetime $M$, the playground for a real scalar field is the trivial line bundle $E=M\times\mathbb{R}$ endowed with the Bosonic bilinear form defined by fiberwise multiplication. As already mentioned before Definition \ref{defSection}, in this case $\Gamma(M,E)\simeq C^\infty(M)$.

\begin{defi}\label{scalar}
A smooth section of $E=M\times\mathbb{R}$, is a real scalar field, if the associated function $\Phi\in C^\infty(M)$ is a solution of the following Cauchy problem: 
\begin{equation}\label{KG}
P\Phi=\left(\Box_\nabla+\xi R+m^2\right)\Phi=0\mbox{ on }M,\quad\Phi=\Phi_0\mbox{ on }\Sigma,
\quad\nabla_\mathfrak{n}\Phi=\Phi_1\mbox{ on }\Sigma,
\end{equation}
where $\Sigma$ is a smooth spacelike Cauchy hypersurface of $M$ with future-directed timelike unit normal vector field $\mathfrak{n}$, whereas $\Phi_0,\Phi_1\in C^\infty_0(\Sigma)$ are given initial data on $\Sigma$. Here $m^2\geq 0$, $R$ is the scalar curvature and $\xi\in\mathbb{R}$, whereas $\Box_\nabla\doteq -g^{\mu\nu}\nabla_\mu\nabla_\nu$ is the d'Alembert operator for the Levi-Civita connection $\nabla$ on $M$.
\end{defi}

The above definition includes in a single framework both massive and massless scalar fields and it encompasses the possibility of a non-trivial coupling to the background geometry. Regardless of the value of $\xi$, we see that, on Minkowski spacetime, we recover both the wave equation ($m^2=0$) and the Klein-Gordon equation ($m^2>0$). Since, at this level, the metric behaves like a background field and hence we are neglecting any backreaction, $\xi R$ has the role of an effective point-dependent mass term. There is no a priori reason to believe that a certain value of $\xi$ is preferred, though $\xi=0$ and $\xi=\frac{1}{6}$ are distinguished, the first representing the so-called {\em minimal coupling} and the second the {\em conformal coupling}. While the reason behind the choice of the word minimal is rather straightforward, the adjective conformal refers to a special property of \eqref{KG}. To wit, let us consider a globally hyperbolic spacetime $(M,g)$ and a conformally rescaled one $(M,g')$, namely, $g'=\Omega^2 g$ up to an isometry, where $\Omega$ is a smooth and strictly positive function on $M$. Then, to every solution $\Phi$ of \eqref{KG} with $m^2=0$ and $\xi=\frac{1}{6}$ on $(M,g)$, one can build $\Phi'=\Omega^{-1}\Phi$, which is a solution of the very same partial differential equation on $(M,g')$.

In a local coordinate system, 
\begin{equation}\label{waveop}
\Box_\nabla=-\frac{1}{\sqrt{|g|}}\partial_\mu\left(g^{\mu\nu}\sqrt{|g|}\partial_\nu\right),
\end{equation}
where $g$ is the determinant of the metric tensor. Hence, according to Definition \ref{nhyp}, $P$ is normally hyperbolic. Therefore we can characterize the space $\mathcal{S}_{sc}(M)$ of spacelike-compact solutions of the real Klein-Gordon equation as being isomorphic to $C^\infty_0(M)/P(C^\infty_0(M))$, according to Lemma \ref{iso}. Applying Stokes theorem, $P$ turns out to be formally self-adjoint, hence, due to Proposition \ref{sympl2}, $\mathcal{S}_{sc}(M)$ is a symplectic space, the symplectic form $\sigma:\mathcal{S}_{sc}(M)\otimes\mathcal{S}_{sc}(M)\to\mathbb{R}$ being given by
\begin{equation}\label{syKG}
\sigma(\Phi_f,\Phi_h)=(f,Gh)=\int_M\,\mathrm{vol}_M\,f\,Gh,
\end{equation}
where $\Phi_f,\Phi_h\in\mathcal{S}_{sc}(M)$ and $f,h\in C^\infty_0(M)$ such that $Gf=\Phi_f,Gh=\Phi_h$.

\subsubsection{Majorana fields}
While scalar fields played and play a prominent role in most of the papers focused on quantum field theory on curved backgrounds in the algebraic approach, for a long time spinor fields have been relegated to an ancillary role. Aside for the work of Dimock\cite{Dimock2}, for many years no paper on this topic was written. In the past five years we have witnessed a resurgence of interest in this topic. Therefore nowadays, our understanding of spinor fields on curved backgrounds is on par with that of scalar fields and several thorough analyses are available, see Refs. \citen{Dappiaggi:2009xj, Fewster:2001js, Sanders, Zahn:2012dz}. In particular, we shall now review how Majorana fields can be introduced on an arbitrary curved background and how their dynamics can be discussed. Notice that the procedure, we shall follow, differs drastically from the standard one in Minkowski spacetime. The requirement of covariance of massive free fields under the action of the Poincar\'e group provides the full set of unitary and irreducible representations induced from the $SU(2)$ subgroup. These are labelled by an integer/half-integer number which ought to be identified with the spin.

On a generic globally hyperbolic spacetime, the potential lack of a non-trivial isometry group forces us to follow a different approach.  In a four dimensional Lorentzian manifold, the starting point is the spin group\cite{Lawson}, the double cover of $SO(3,1)$. Its component connected to the identity is, therefore, isomorphic to $SL(2,\mathbb{C})$. Furthermore, to each point $x$ of an oriented and time-oriented differentiable Lorentzian manifold, we can assign an orthonormal oriented and time-oriented basis of $T_xM$. 
We denote the collection of all these frames at $x$ as $F_xM$. 
Since two elements of $F_xM$ can be always mapped into each other by the action of an element of the proper orthochronous Lorentz group, by fixing an element, we can identify $F_xM$ with the Lie group $SO_0(3,1)$. Since the tangent space of a globally hyperbolic spacetime $M$ is trivial, that is $TM\simeq M\times\mathbb{R}^4$, it turns out that $FM$, the disjoint union $\bigsqcup_{x\in M}F_xM$, is isomorphic to $M\times SO_0(3,1)$ as a principal bundle. $FM$ is also called the {\it bundle of Lorentz frames} and it is a principal $SO_0(3,1)$-bundle, see Ref. \citen{Husemoller} for the definition. Since we mention principal bundles only in this section and all those we need turn out to be trivial, we will skip the detailed mathematical analysis.

\begin{defi}\label{spinstr}
Given a spacetime $M$, a {\it spin structure} is a pair $(SM, \rho)$, 
where $SM$ is a principal $SL(2,\mathbb{C})$-bundle (in particular, each fiber $S_xM$ is 
isomorphic to $SL(2,\mathbb{C})$) and $\rho:SM\to FM$ is a principal bundle map, 
namely a smooth map fulfilling the following conditions:
\begin{enumerate}
\item $\rho$ preserves the fibers, that is to say that the image of $S_xM$ lies in $F_xM$; 
\item $\rho$ is equivariant, {\it i.e.}, using $\Pi$ to denote the surjective homomorphism from 
$SL(2,\mathbb{C})$ to $SO_0(3,1)$, for every $A\in SL(2,\mathbb{C})$, 
$\rho\circ R_{A}=R_{\Pi(A)}\circ\rho$, where $R_A$ and $R_{\Pi(A)}$ are the right Lie group actions 
of $A\in SL(2,\mathbb{C})$ on $SM$ and, respectively, of $\Pi(A)\in SO_0(3,1)$ on $FM$.
\end{enumerate}
\end{defi}

Notice that, as for $FM$, also $SM$, the {\it spin bundle} is defined using the language of principal bundles. Nonetheless, if $M$ is globally hyperbolic, one can prove\footnote{We are grateful to Chris Fewster for pointing us out Ref. \citen{Isham:1978ec}.} that $M$ admits only the trivial principal $SL(2,\mathbb{C})$-bundle\cite{Isham:1978ec}. 
This means that, as a matter of fact, in Definition \ref{spinstr} we are just taking into account principal bundles 
which look like $M\times SL(2,\mathbb{C})$. 
The existence of a spin structure on a globally hyperbolic spacetime is thus straightforward, 
all principal bundles being trivial. This fact agrees with a more general result. 
To wit, it is proven in Ref. \citen{Borel}, that an oriented, time oriented, differentiable manifold $M$ admits a spin structure if and only if $w_2(M)$, its second Stiefel-Whitney class\cite{Husemoller}, is trivial. As shown for example in Lemma 2.1 in Ref. \citen{Dappiaggi:2009xj}, this is the case for every globally hyperbolic spacetime. It is important to stress that, nonetheless, the spin structure is not necessarily unique, the number of inequivalent possibilities\cite{Lawson} being ruled by $H^1(M,\mathbb{Z}_2)$, the first cohomology group with $\mathbb{Z}_2$ coefficients.

The existence of a spin structure and of the representations\cite{BR} of $SL(2,\mathbb{C})$ 
suggests the following definition:

\begin{defi}\label{defDiracBundle}
We call {\it Dirac bundle} of a globally hyperbolic spacetime $M$ with respect to the 
representation $T\doteq D^{(\frac{1}{2},0)}\oplus D^{(0,\frac{1}{2})}$ of $SL(2,\mathbb{C})$ 
on $\mathbb{C}^4$ the associated bundle $DM\doteq SM\times_T\mathbb{C}^4$, 
which is defined as the orbit space of $SM\times\mathbb{C}^4$ under the 
right group action of $SL(2,\mathbb{C})$ induced by $T$. This means that an element 
$[(p,A)]$ in $DM$, $p\in SM, A\in SL(2,\mathbb{C})$, is an equivalence class with respect to the 
following relation: $(p',z')\sim(p,z)$ if and only if there exists $A\in SL(2,\mathbb{C})$ 
such that $p'=R_A(p)$ and $z'=T(A^{-1})z$. 
Since $SM$ is a trivial bundle, $DM$ is isomorphic to $M\times\mathbb{C}^4$ as a vector bundle. 
A {\it Dirac field/spinor} $\psi$ is a smooth section of $DM$ and 
thus it can be read as a smooth map from $M$ to $\mathbb{C}^4$.
\end{defi}

In order to define Majorana fields and to write down the Dirac equation we need the notion of $\gamma$-matrices and charge conjugation, we use here the conventions of Appendix A in Ref. \citen{Hack:2012dm}. The $\gamma$-matrices $\gamma^a$, $a=0,\dots,3$, 
are complex $(4\times 4)$-matrices satisfying the Clifford algebra relations
$\{\gamma^a,\gamma^b\}=2\,\eta^{ab}$, $\eta$ being the Minkowski metric with $(-,+,+,+)$ on the diagonal. 
We take the timelike $\gamma$-matrix to be 
antihermitian, ${\gamma^0}^\dagger =-\gamma^0$, and the spatial $\gamma$-matrices 
hermitian, ${\gamma^i}^\dagger= \gamma^i$, for all $i=1,2,3$. We further fix $\beta := i\gamma^0$
which satisfies $\beta^\dagger =\beta$. There exists a charge conjugation matrix $C$, which is antisymmetric, {\it i.e.}~$C^{\mathrm{T}}=-C$. Further properties are $C^\dagger = C^{-1}$ and, for all $a=0,\dots,3$, ${\gamma^a}^{\mathrm{T}} = - C\gamma^a C^{-1}$. We define the charge conjugation operation on spinors $\chi\in\mathbb{C}^4$ by $\chi^c \doteq -\beta\,\overline{C}\, \overline{\chi}$,
where  $\overline{\cdot}$ denotes component-wise complex conjugation. This operation squares to the identity,
${\chi^c}^c =\chi$, for all $\chi$. A Majorana spinor is defined by the reality condition $\chi^c =\chi$
and the space of Majorana spinors is a real vector space of dimension $4$. We define a non-degenerate bilinear form on spinors $\chi\in\mathbb{C}^4$ by $\langle\langle \chi_1, \chi_2\rangle\rangle \doteq  -i\chi^\dagger_1\beta\chi_2$. 
For every Majorana spinor $\chi$ the Dirac adjoint equals the Majorana adjoint, 
$\chi^\dagger \beta = \chi^{\mathrm{T}} C$, and thus the bilinear form can be equivalently expressed as
$$\langle\langle \chi_1, \chi_2\rangle\rangle=-i\chi^\text{T}_1 C \chi_2$$ on Majorana spinors. From this it is readily seen that $\langle\langle \cdot, \cdot\rangle\rangle$ is a non-degenerate, real-valued, skew-symmetric bilinear form on Majorana spinors.

\begin{defi}\label{defMajoranaBundle}
The {\it Majorana bundle} $DM_\mathbb{R}$ of a globally hyperbolic spacetime $M$ is the subbundle of $DM$ defined by $DM_\mathbb{R}\doteq\{p\in DM \,| \,p^c=p\}$ where $\cdot^c$ is the fibrewise lift of the charge conjugation map on $\mathbb{C}^4$-spinors to $DM$. We endow $DM_\mathbb{R}$ with the fibrewise Fermionic non-degenerate bilinear form $\langle\cdot,\cdot\rangle$ induced by the bilinear form $\langle\langle\cdot,\cdot\rangle\rangle$ on $\mathbb{C}^4$.
\end{defi}

We are now ready to sketch the definition of the ingredients 
which will enter the Dirac equation on a curved background. 
For more details refer to Section 2 in Ref. \citen{Dappiaggi:2009xj}. First we lift the usual ``flat spacetime'' $\gamma$-matrices to ``curved spacetimes'' ones, {\it i.e.}~to sections of $TM\otimes \mathrm{Hom}(DM,DM)$ by means of frames. In order to write down the Dirac equation we introduce a covariant derivative for sections of $DM$ by using $\rho$ to pull back to $SM$ the Levi-Civita connection, 
which is a principal bundle connection on $FM$. This yields an induced covariant derivative on $DM$
which will be denoted as $\nabla$. Given a choice of a spacelike smooth Cauchy hypersurface $\Sigma$, 
the dynamics for Majorana spinors $\psi\in\Gamma(DM_\mathbb{R})$ is ruled by the following initial value problem: 
\begin{eqnarray}
D\psi\doteq-\gamma^\mu\nabla_\mu\psi+m\psi=0\mbox{ on }M, & \quad 
& \psi=\psi_0\mbox{ on }\Sigma, \label{Dirac}
\end{eqnarray}
where $\psi_0$ is a compactly supported section of $DM_\mathbb{R}$, whereas $m$ is a real number. Note that for all $\psi\in \Gamma(DM)$, $(D\psi)^c=D\psi^c$, thus the Dirac equation is indeed well-defined also for Majorana spinors.

One can infer at first glance that the operator $D$ cannot be normally hyperbolic 
since it is a first order operator. Nonetheless one can show that it is Green-hyperbolic, 
according to Definition \ref{Ghyp}. The procedure to show this fact is based on the introduction 
of the ancillary operator $\widetilde{D}=\gamma^\mu\nabla_\mu+m$. 
As a matter of fact the following theorem holds true:

\begin{theo}
It holds that $D\circ\widetilde{D}=\widetilde{D}\circ D=\Box_\nabla+m^2+R/4\doteq P$ 
where $\Box_\nabla$ is the normally hyperbolic operator constructed out of the spin connection $\nabla$ 
as in Proposition \ref{propBoxDecomposition} and $R$ is the scalar curvature.

Given a smooth spacelike Cauchy hypersurface $\Sigma$ and 
denoting with $\mathfrak{n}$ the timelike unit normal vector field on $\Sigma$, 
each solution $\psi\in\Gamma(DM_\mathbb{R})$ of the Cauchy problem \eqref{Dirac} can be written as $\psi=\widetilde{D}u$, 
$u\in\Gamma(DM_\mathbb{R})$ being a solution of the Cauchy problem $Pu=0$ on $M$, $u=0$ on $\Sigma$ 
and $\nabla_\mathfrak{n}u=\gamma^\mu n_\mu\psi_0$ on $\Sigma$. 

Furthermore $D$ is Green-hyperbolic and formally self-adjoint, 
namely $(D u,v)=(u,Dv)$ holds true for all $u,v\in\Gamma(DM_\mathbb{R})$ such that $\mathrm{supp}(u)\cap\mathrm{supp}(v)$ is compact. 
Hence we recover all the results of Subsection \ref{subsectGreenOp}.
\end{theo}

\begin{proof}
For the first part of the theorem we refer the interested reader to Lemma 2.4 in Ref. \citen{Dappiaggi:2009xj}, where it is proven for arbitrary Dirac spinors and not only for Majorana ones, 
while we focus the attention on the last statement.

First of all $(D u,v)=(u,Dv)$ follows via Stokes' theorem because the covariant derivative on spinors is per construction compatible with the Levi-Civita connection, 
namely $\nabla\langle u,v\rangle=\langle\nabla u,v\rangle+\langle u,\nabla v\rangle$, 
and because $\langle \gamma^\mu u, v \rangle = -\langle u, \gamma^\mu v \rangle$, 
which follows from ${\gamma^a}^{\mathrm{T}} = - C\gamma^a C^{-1}$. 
This argument shows that $\widetilde{D}$ is formally self-adjoint too.

We show how to construct advanced and retarded Green operators for $D$. 
By direct inspection we realize that $P$ is normally hyperbolic and formally self-adjoint. 
Hence, according to Theorem \ref{thmGreenOp}, 
it has unique advanced and retarded Green operators $G_P^+,G_P^-$ 
fulfilling $(G_P^\pm u,v)=(u,G_P^\mp v)$ for all $u,v\in\Gamma_0(M,E)$. 
We define $G^\pm=G_P^\pm\circ\widetilde{D}$ 
and we claim that these are advanced and retarded Green operators for $D$. 
In fact, for each $u\in\Gamma_0(DM)$, $G^\pm Du=G_P^\pm Pu=u$ 
and the support of $G^\pm u$ is included in the causal future/past of 
$\mathrm{supp}(\widetilde{D}u)\subseteq\mathrm{supp}(u)$. 
$DG^\pm u=u$ is still to be checked for each $u\in\Gamma_0(DM)$. 
To this scope, considering $u,v\in\Gamma_0(M,E)$, we deduce the following chain of identities:
\begin{equation}
\begin{array}{rcl}
(v,DG^\pm u) & = & (PG_P^\mp v,DG^\pm u) \\
& = & (DG_P^\mp v,PG^\pm u) \\
& = & ((DG_P^\mp v,\widetilde{D}u) \\
& = & ((PG_P^\mp v,u) \\
& = & (v,u).
\end{array}
\end{equation}
Non-degeneracy of $(\cdot,\cdot)$ and arbitrariness of $v$ entail the thesis.
\end{proof}

As the Dirac operator $D$ is formally self-adjoint and Green-hyperbolic, applying Lemma \ref{iso}, we can characterize the space $\mathcal{S}_{sc}(M)$ of spacelike-compact Majorana solutions of the Dirac equation as being isomorphic to $\Gamma_0(DM_\mathbb{R})/D(\Gamma_0(DM_\mathbb{R}))$. Due to Proposition \ref{sympl2}, we can endow $\mathcal{S}_{sc}(M)$ with an inner product, {\it i.e.}~$\sigma:\mathcal{S}_{sc}(M)\otimes\mathcal{S}_{sc}(M)\to\mathbb{R}$ defined by
\begin{equation}\label{syDirac}
\sigma(\Psi_f,\Psi_h)=(f,Gh)=\int_M\,\mathrm{vol}_M\langle f,Gh\rangle,
\end{equation}
where $\Psi_f,\Psi_h\in\mathcal{S}_{sc}(M)$ and $f,h\in \Gamma_0(DM_\mathbb{R})$ such that $Gf=\Psi_f,Gh=\Psi_h$.

\subsubsection{Proca fields}

As a last example, we consider massive spin $1$ fields. In the literature, this has been discussed by a few authors\cite{Furlani:1999kq, Fewster:2003ey, Dappiaggi:2011ms}, mostly because it is the simplest example of a vector boson. In this case the underlying vector bundle is the cotangent space $T^\ast M$. Since we consider only those $M$ which are globally hyperbolic spacetimes, this entails that $TM$ is trivial and thus $T^\ast M$ too, that is $T^*M\simeq M\times\mathbb{R}^4$. We endow this vector bundle with the Bosonic bilinear form induced by the (inverse) metric.

\begin{defi}\label{Proca}
Let $M$ be a globally hyperbolic spacetime. A {\it Proca field} is a section $A\in\Gamma(T^\ast M)=\Omega^1(M)$, that is a differential $1$-form, fulfilling the following equation of motion:
\begin{equation}\label{Proca1}
PA\doteq\delta dA+m^2 A=0,
\end{equation}
where $m^2>0$, $d:\Omega^k(M)\to\Omega^{k+1}(M)$ is the {\it exterior derivative}, whereas $\delta\doteq (-1)^k\ast^{-1}d\ast:\Omega^k(M)\to\Omega^{k-1}(M)$ is the {\it codifferential} defined out of the {\it Hodge dual} $\ast$. 
\end{defi} 

Contrary to \eqref{KG}, it is not manifest at first glance that the dynamics of a Proca field is ruled by a normally hyperbolic operator. Nonetheless, if we remember that $d^2=0$, hence $\delta^2=0$ too, applying the codifferential to the equation of motion, one obtains $\delta A=0$. In other words every solution of \eqref{Proca1}, also solves $(\Box+m^2)A=0$, where the operator $\Box\doteq\delta d+d\delta$ is the so-called {\it Laplace-de Rham wave operator} and it coincides with $\Box_\nabla$ defined in Proposition \ref{propBoxDecomposition}, $\nabla$ being the Levi-Civita connection acting on sections of $T^\ast M$. In other words a Proca field satisfies a normally hyperbolic equation. More precisely, we can translate Definition \ref{Proca} into a normally hyperbolic equation together with a constraint:
\begin{equation}\label{CauchyProca}
\widetilde{P}A\doteq\Box A + m^2 A=0,\quad\delta A=0.
\end{equation}
Although $\widetilde{P}=\Box+m^2$ is normally hyperbolic and thus it possesses unique advanced and retarded Green operators $\widetilde{G}^+,\widetilde{G}^-$, we need to cope with the additional constraint $\delta A=0$ in order to generate the full space of solutions with spacelike-compact support. This result is obtained in the following theorem, which exploits a trick also used in Refs. \citen{Furlani:1999kq, Fewster:2003ey} in order to define advanced and retarded Green operators for $P$. 

Before we proceed to the construction of advanced and retarded Green operators for $P$, 
let us specify the Bosonic non-degenerate bilinear form we consider for the Proca field. 
This can be defined on the bundle $\bigwedge^k(T^\ast M)$ of $k$-forms 
as $\langle\cdot,\cdot\rangle=\ast^{-1}(\cdot\wedge\ast\,\cdot)$. 
For each $p\in M$ and for each $\alpha,\beta\in\bigwedge^k(T_p^\ast M)$, 
this simply reads $\langle\alpha,\beta\rangle=g^{\mu_1\nu_1}(p)\cdots 
g^{\mu_k\nu_k}(p)\alpha_{\mu_1\dots\mu_k}\beta_{\nu_1\dots\nu_k}$, 
where we used local coordiantes at $p$ on the right hand side. 
As usual, $\langle\cdot,\cdot\rangle$ induces a non-degenerate pairing $(\cdot,\cdot)$ between sections. 
To construct the symplectic form of the Proca field, we need this operation only for $k=1$. 
Nonetheless, the following proof will involve forms of different order. 
This is the reason that induced us to present the definition for arbitrary $k$.

\begin{theo}
$P$ is Green-hyperbolic and formally self-adjoint, 
namely $(\omega,P\omega^\prime)=(P\omega,\omega^\prime)$ 
for each $\omega,\omega^\prime\in\Omega^1(M)$ 
such that $\mathrm{supp}(\omega)\cap\mathrm{supp}(\omega^\prime)$ is compact.
Hence all conclusions of Subsection \ref{subsectGreenOp} hold true for $P$. 
\end{theo}

\begin{proof}
We have only to show that $P$ is formally self-adjoint and Green hyperbolic. 
Formal self-adjointness follows from Stokes' theorem: 
For arbitrary $\omega\in\Omega^{k-1}(M)$ and $\omega^\prime\in\Omega^{k}(M)$ we have
\begin{equation}
d(\omega\wedge\ast\omega^\prime)=d\omega\wedge\ast\omega^\prime+(-1)^{k-1}\omega\wedge d\ast\omega^\prime=d\omega\wedge\ast\omega^\prime-\omega\wedge\ast\delta\omega^\prime.
\end{equation}
If we assume that the intersection of the supports of $\omega$ and $\omega^\prime$ is compact, 
integrating the last equation, we read $0=(d\omega,\omega^\prime)-(\omega,\delta\omega^\prime)$. 
This shows that $d$ and $\delta$ are formal adjoints of each other, 
in particular $P=\delta d+m^2$ turns out to be formally self-adjoint.

As already mentioned above, $\widetilde{P}$ is normally hyperbolic, 
hence, according to Theorem \ref{thmGreenOp}, 
there is a unique advanced/retarded causal propagator $\widetilde{G}^\pm$ for $\widetilde{P}$. 
In the rest of the proof for convenience we denote $(\mathrm{id}_{\Omega^1_0(M)}+m^{-2}d\delta)$ with $L$. 
We claim that $G^\pm=L\circ\widetilde{G}^\pm$ is an advanced/retarded Green operator for $P$: 
The support property for Green operators is automatically implemented by this formula. 
On account of nilpotency of $d$, for $\omega\in\Omega_0^1(M)$ 
we have $PG^\pm\omega=\widetilde{P}\widetilde{G}^\pm\omega=\omega$. 
We have still to check $G^\pm P\omega=\omega$. 
To prove this, we first note that the argument showing formal self-adjointness of $P$ 
entails the same property for $\widetilde{P}$ too. 
Moreover, from $d^2=0$ and $\delta^2=0$ 
we deduce that $\widetilde{P}$ commutes with $L=\mathrm{id}_{\Omega^1_0(M)}+m^{-2}d\delta$ 
and that $L\circ P=\widetilde{P}$. Then, taking $\omega,\omega^\prime\in\Omega^1_0(M)$, we have
\begin{equation}
\begin{array}{rcl}
(\omega^\prime,G^\pm P\omega) & = & (\widetilde{P}\widetilde{G}^\mp\omega^\prime,G^\pm P\omega) \\
& = & (\widetilde{G}^\mp\omega^\prime,\widetilde{P}G^\pm P\omega) \\
& = & (\widetilde{G}^\mp\omega^\prime,L\widetilde{P}\widetilde{G}^\pm P\omega) \\
& = & (\widetilde{G}^\mp\omega^\prime,\widetilde{P}\omega) \\
& = & (\omega^\prime,\omega).
\end{array}
\end{equation}
Non degeneracy of $(\cdot,\cdot)$, together with the freedom in the choice of $\omega^\prime$, entails the thesis.
\end{proof}
The last theorem entails that 
all the results of Subsection \ref{subsectGreenOp} are recovered in the case of the Proca field, 
meaning that the space $\mathcal{S}_{sc}(M)$ of spacelike compact solutions of the Proca equation 
is isomorphic to the quotient space $\Omega^1_0(M)/P(\Omega_0^1(M))$. 
Moreover Proposition \ref{sympl2} provides the symplectic form for the Proca field, 
namely $\sigma:\mathcal{S}_{sc}\otimes\mathcal{S}_{sc}\to\mathbb{R}$ defined by
\begin{equation}
\sigma(A_\omega,A_{\omega^\prime})=(\omega,G\omega^\prime)
=\int_M\omega\wedge\ast G\omega^\prime,
\end{equation}
where $A_\omega,A_{\omega^\prime}\in\mathcal{S}_{sc}(M)$ 
and $\omega,\omega^\prime\in\Omega^1_0(M)$ such that 
$G\omega=A_\omega,G\omega^\prime=A_{\omega^\prime}$.\\

\noindent Before concluding the section, we stress that we have discussed only the simplest possible examples of free fields whose classical dynamics is ruled by a Green-hyperbolic operator. In the case of spin $1$ we considered, for example, only the massive case, since the massless one is related to Maxwell's equations whose dynamics is ruled by a Green-hyperbolic operator only if one exploits gauge invariance. At a quantum level this becomes a rather complicated topic which is still hotly debated. In this review we will not enter into the analysis of the role of gauge invariance and we suggest an interested reader to refer to Ref. \citen{Hack:2012dm} for a thorough analysis of its connections to Green-hyperbolic operators. For electromagnetism, one can consult Refs. \citen{Benini:2013tra, Ciolli:2013pta, Dappiaggi:2011cj, Dappiaggi:2011zs, Dappiaggi:2011ms, Dimock:1992ff, Fewster:2003ey, Pfenning:2009nx, Sanders:2012sf}. For the case of spin $3/2$ fields, there are several obstructions to their quantization on a generic globally hyperbolic spacetime and, for this reason, we avoid discussing the associated classical dynamics. We refer to Ref. \citen{Hack:2011yv} and also to Refs. \citen{Hack:2012dm, Schenkel:2011nv} for a critical survey and for an analysis of the connection to supergravity theories. On the contrary spin $2$ fields have been for long neglected and only recently they have been discussed in the literature, see Ref. \citen{Fewster:2012bj}.

\section{Quantization: Algebra and States}
Goal of this section is to discuss the quantization of the free field theories, whose classical dynamics has been developed in the previous section. We shall work within the framework of algebraic quantum field theory as originally envisaged by Haag and Kastler in Ref. \citen{Haag:1963dh}. First formulated under the assumption that the underlying background is Minkowski spacetime, such approach is based on a two-step procedure. The first consists of associating to a physical system a suitable topological algebra $\mathcal{A}$ together with a map $*:\mathcal{A}\to\mathcal{A}$ such that $*\circ *$ coincides with the identity. The resulting {\em $*$-algebra} is an abstract realization of the observables of the underlying physical system. The second step, instead, is based on the choice of a state, namely a continuous positive functional on $\mathcal{A}$ which allows to represent $\mathcal{A}$ itself in terms of suitable linear operators on a Hilbert space. In the forthcoming discussion we will make mathematically precise these statements, we will show how the algebra $\mathcal{A}$ is explicitly constructed for a given free field theory and we will introduce both the notion of an algebraic state and the constraints it has to satisfy in order to be physically sensible. 

\subsection{The field algebra} 
As we have mentioned, the seminal paper of Haag and Kastler did not elaborate on the possibility to quantize a field theory on a curved background. The first axiomatic, algebraic formulation of a quantum theory on a curved background is due to Dimock\cite{Dimock}. He focused his attention on extending the core of Ref. \citen{Haag:1963dh} by writing down a set of axioms which any sensible algebra of observables should satisfy regardless of the details of the background metric. We will not report them as in the original paper. The reason is of technical nature: It is required that one associates to a field theory a $C^*$-algebra, see for example Ref. \citen{Haag:1992hx} for the definition. Such hypothesis is advantageous particularly because, after choosing a state, the algebra can be represented in terms of bounded linear operators on a Hilbert space. In our case, we will be interested in constructing an algebra of observables which can be extended so to include also the curved background counterpart of the normal ordered field, used on Minkowski spacetime to study interactions in a perturbative regime. For this reason the operators, we will be considering, are not bounded and, hence, the underlying structure is that of a $*$-algebra. Therefore, for a given globally hyperbolic spacetime $(M,g)$, we call {\it algebra of local observables} any algebra $\mathcal{A}(M)$ such that the following axioms are fulfilled:
\begin{enumerate}
\item To every contractible open bounded set $\mathcal{O}\subseteq M$ one associates a $*$-algebra $\mathcal{A}(\mathcal{O})$. Such assignment is {\it isotonous}, that is, if $\mathcal{O}\subseteq\mathcal{O}'$ then $\mathcal{A}(\mathcal{O})\subseteq\mathcal{A}(\mathcal{O}')$. The $*$-algebra of local observables $\mathcal{A}(M)$ is defined as the union of all $\mathcal{A}(\mathcal{O})$ with $\mathcal{O}\subseteq M$ contractible open bounded; 
\item If $\mathcal{O}$ and $\mathcal{O}'$ are causally separated, that is $\mathcal{O}\cap J_M(\mathcal{O}')=\emptyset$, then $[a,b]=0$ for all $a\in\mathcal{A}(\mathcal{O})$ and for all $b\in\mathcal{A}(\mathcal{O}')$, where the commutator is evaluated in the full algebra $\mathcal{A}(\mathcal{M})$; 
\item For any isometry of $M$, that is a diffeomorphism $\iota:M\to M$ such that $\iota^*g=g$, there exists an isomorphism $\alpha_\iota:\mathcal{A}(M)\to\mathcal{A}(M)$ for which $\alpha_\iota(\mathcal{A}(\mathcal{O}))=\mathcal{A}(\iota(\mathcal{O}))$, for all open bounded set $\mathcal{O}\subseteq M$. Furthermore, if $\iota$ is the identity map, so is $\alpha_\iota$, whereas, if we consider two isometries $\iota$ and $\iota'$, $\alpha_\iota\circ\alpha_{\iota'}=\alpha_{\iota\circ\iota'}$.
\end{enumerate}
It is noteworthy that the second axiom implements the property of causality in $\mathcal{A}(M)$ since it ensures that observables which are localized in spacetime regions which are causally disconnected are commuting. The third axiom guarantees, instead, the compatibility between $\mathcal{A}(M)$ and the symmetries of the background. It translates on a curved background the standard requirement of covariance under the action of the Poincar\'e group, which is at the heart of any textbook about quantum field theories on Minkowski spacetime. 

One of the biggest advantages of the given set of axioms is the possibility to show the existence of a concrete construction of an algebra of local observables for free fields. We recall that a Bosonic (Fermionic) classical free field theory is (in the case of a Hermitian field) completely specified in terms of a real vector bundle $E$ endowed with a Bosonic (Fermionic) non-degenerate bilinear form and of a partial differential operator $P$ which is formally self-adjoint with respect to this bilinear form and furthermore Green-hyperbolic. We can straightforwardly construct an algebra which encodes the simplest observables of the associated quantized field theories as follows:

\begin{defi}\label{fieldal}
We call {\it field algebra} $\mathcal{F}(M)$ of a Bosonic or Fermionic field $\Phi$ specified by the vector bundle $E$ endowed with the Bosonic or Fermionic bilinear form $\langle \cdot, \cdot\rangle$ and the Green-hyperbolic formally self-adjoint operator $P$ the $*$-algebra freely generated by a unit $\mathbb{I}$ and by the symbols $\widehat{\Phi}(f)$, $f\in \Gamma_0(M,E)_\mathbb{C}\doteq\Gamma_0(M,E)\otimes_\mathbb{R}\mathbb{C}$, from which we single out the ideal generated by the following relations:
\begin{itemize}
\item $\widehat{\Phi}(k_1f_1+k_2f_2)=k_1\widehat{\Phi}(f_1)+k_2\widehat{\Phi}(f_2)$ for all $k_1,k_2\in\mathbb{C}$ and for all $f_1,f_2\in \Gamma_0(M,E)_\mathbb{C}$;
\item $\widehat{\Phi}(Pf)=0$ for all $f\in \Gamma_0(M,E)_\mathbb{C}$;
\item $\widehat{\Phi}(f)^*=\widehat{\Phi}(\overline{f})$, for all $f\in \Gamma_0(M,E)_\mathbb{C}$, where $\overline{\cdot}$ denotes complex conjugation;
\item $[\widehat{\Phi}(f_1),\widehat{\Phi}(f_2)]_\mp\doteq\widehat{\Phi}(f_1)\widehat{\Phi}(f_2)\mp\widehat{\Phi}(f_2)\widehat{\Phi}(f_1)=i\sigma(Gf_1,Gf_2)\mathbb{I}=i(f_1,Gf_2)\mathbb{I}$ for all $f_1,f_2\in \Gamma_0(M,E)_\mathbb{C}$, where $(\cdot,\cdot)$ is the bilinear form on sections of $E$ induced by $\langle\cdot,\cdot\rangle$. Here, the minus (plus) sign applies in the Bosonic (Fermionic) case.
\end{itemize}
\end{defi}

Note that by Proposition \ref{sympl2} the signs in the canonical (anti)commutation relations $[\widehat{\Phi}(f_1),\widehat{\Phi}(f_2)]_\mp=i\sigma(Gf_1,Gf_2)$ are consistent. The symbols $\widehat{\Phi}(f)$ can be interpreted as quantizations of the classical observables (functionals) on smooth solutions $\mathcal{S}(M)$ of $P\Phi=0$ defined by $\mathcal{O}_f:\mathcal{S}(M)\ni \Phi\mapsto (f,\Phi)$. Thus, interpreting the ``smeared field'' generators $\widehat{\Phi}(f)$ as $(f,\widehat\Phi)$, the above algebraic relations can be rephrased as follows for ``unsmeared fields'' $\widehat\Phi(x)$:

\begin{itemize}
\item $P\widehat{\Phi}(x)=0$;
\item $\widehat{\Phi}(x)^*=\widehat{\Phi}(x)$;
\item $[\widehat{\Phi}(x),\widehat{\Phi}(y)]_\mp=iG(x,y)\mathbb{I}$, where $G(x,y)$ is the bidistribution defined by $$\int_M\mathrm{vol}_M(x)\int_M\mathrm{vol}_M(y)\; f(x)^\dagger(G(x,y)h(y))\doteq(f,Gh)$$
for all $f,h\in\Gamma_0(M,E)$ and $f^\dagger\in\Gamma_0(M,E^\ast)$ is here defined by $f^\dagger(v)\doteq\langle f, v\rangle$ for all $v\in E$.
\end{itemize}

A potential obstruction arises in the case of Fermionic theories. In the next section we shall discuss algebraic states on $\mathcal{F}(M)$, which correspond to representations of the symbols $\widehat{\Phi}(f)$ as linear operators on a Hilbert space. A necessary condition for such states/representations to exist (``unitarity condition'') in the Fermionic case is that $i\sigma(Gf,Gf)\geq 0$ for all  $f\in\Gamma_0(M,E)$ because for $f_1=f_2=f\in\Gamma_0(M,E)$ the left hand side of the anticommutation relations is equal to $2 \widehat{\Phi}(f)^2$, which is a positive operator in any Hilbert space representation. Fortunately, in the case of a Majorana field it is not difficult to prove that this consistency condition is met. See {\it e.g.}~Refs. \citen{Dappiaggi:2009xj, Sanders} for this and further details on the quantization of Majorana and Dirac fields in the algebraic language.

It is easy to verify that $\mathcal{F}(M)$ gives rise to a local algebra of observables. For every contractible open bounded set $\mathcal{O}\subseteq M$, let $\mathcal{F}(\mathcal{O})$ be the subalgebra of $\mathcal{F}(M)$ generated by those $f\in \Gamma_0(M,E)_\mathbb{C}$ such that $\supp(f)\subseteq\mathcal{O}$. Isotony is then guaranteed per construction. The axiom of locality is a by-product, instead, of the commutation relations between the generators. As a matter of facts, assuming that $\mathcal{O},\mathcal{O}'$ are causally separated, we deduce $[\widehat{\Phi}(f_1),\widehat{\Phi}(f_2)]_\mp=i(f_1,Gf_2)=0,$
whenever $\supp(f_1)\subseteq\mathcal{O}$ and $\supp(f_2)\subseteq\mathcal{O}'$. In the Bosonic case this is already sufficient, whereas in the Fermionic case single fields only anticommute, as they are not observable. However, it is straightforwardly shown that {\it e.g.}~elements of the subalgebra of $\mathcal{F}(M)$ generated by even powers of single fields commute at spacelike separations.

Eventually, if we consider an isometry $\iota:M\to M$, we can construct a map $\alpha_\iota:\mathcal{F}(M)\to\mathcal{F}(M)$ by defining it on the generators, that is $\alpha_\iota(\mathbb{I})=\mathbb{I}$ and $\alpha_\iota(\widehat{\Phi}(f))=\widehat{\Phi}(f\circ\iota^{-1})$ for all $f\in \Gamma_0(M,E)_\mathbb{C}$. The details of the proof that $\alpha_\iota$ is indeed an isomorphism are left to the reader. 

Let us stress that, when $M$ is Minkowski spacetime, the procedure presented above coincides with the usual quantization of Bosonic (Fermionic) free field theories via canonical (anti)commutation relations. The advantage of this presentation relies on the fact that it can be directly applied without any change to general globally hyperbolic spacetimes.

\subsection{Algebraic states and Hilbert space representations}

In the last part of this review we introduce the notion of algebraic state and the criteria which allow us to select those which are physically sensible. As a starting point we consider any topological $*$-algebra $\mathcal{A}$ with a unit element $\mathbb{I}$. $\mathcal{A}$ is not necessarily the field algebra of Definition \ref{fieldal}, although of course this represents our primary example.

\begin{defi}\label{state}
An algebraic {\it state} is a continuous linear functional $\omega:\mathcal{A}(M)\to\mathbb{C}$ such that 
$$\omega(\mathbb{I})=1,\qquad\omega(a^*a)\geq 0,\;\forall a\in\mathcal{A}(M).$$
\end{defi}

A functional fulfilling the first condition is said to be {\it normalized}, whereas, when the second holds true, it is called {\it positive}. Notice that, in concrete examples, the condition, which is usually rather complicated to check, is the second one since it is highly non-linear. For those, who are used to the standard approach to quantum field theory on Minkowski background, it is at first glance hard to believe that the assignment of a pair $(\mathcal{A},\omega)$ is indeed tantamount to quantizing a local algebra of observables. Since this is a key step in the algebraic approach to quantum field theory, we shall discuss this point thoroughly. In particular we shall prove the Gelfand-Naimark-Segal (GNS) theorem. Our analysis will follow closely the script of Klaus Fredenhagen\cite{Fredenhagen}. As a starting point, we show that, whenever we represent a $*$-algebra on a Hilbert space via linear operators, we can automatically construct several states:

\begin{lem}
Let $\mathcal{A}$ be any topological $*$-algebra with an identity element and $\mathcal{H}$ a Hilbert space with scalar product $(\cdot,\cdot)$, such that there exists a faithful strongly continuous representation $\pi:\mathcal{A}\to\mathcal{L}(\mathcal{D})$, where $\mathcal{D}$ is a dense subspace of $\mathcal{H}$, $\mathcal{L}(\mathcal{D})$ is the space of continuous linear operators on $\mathcal{D}$ and where $\pi(a^*)=\pi(a)^*$ for all $a\in\mathcal{A}$. Then, for any $\psi\in\mathcal{D}$ of unit norm, the functional $\omega_\psi:\mathcal{A}\to\mathbb{C}$ defined by $\omega_\psi(a)\doteq(\psi,\pi(a)\psi)$ is a state on $\mathcal{A}$.
\end{lem}

\begin{proof}
Let $\psi\in\mathcal{D}$ be any element such that $\|\psi\|_{\mathcal{H}}=1$. 
Let $\omega_\psi(a)\doteq (\psi,\pi(a)\psi)$. Per construction $\omega_\psi$ is linear and continuous 
since $\pi$ is linear and strongly continuous. 
$\omega_\psi(\mathbb{I})=1$ follows from $\|\psi\|_{\mathcal{H}}=1$ 
and $\pi(\mathbb{I})=\mathrm{id}_\mathcal{D}$, $\pi$ being a representation. 
To conclude we notice that 
\begin{equation*}
\omega_\psi(a^*a)\doteq (\psi,\pi(a^*a)\psi)=(\psi,\pi(a)^*\pi(a)\psi)=\|\pi(a)\psi\|_{\mathcal{H}}^2\geq0,
\end{equation*}
where we exploited $\pi(ab)=\pi(a)\pi(b)$ and $\pi(a^*)=\pi(a)^*$.
\end{proof}

Actually, we can even prove that every state on a topological $*$-algebra $\mathcal{A}$ with a unit element induces a Hilbert space representation of $\mathcal{A}$.

\begin{theo}
Let $\omega$ be a state on a topological $*$-algebra $\mathcal{A}$ with a unit element. 
There exists a dense subspace $\mathcal{D}$ of a Hilbert space $(\mathcal{H},(\cdot,\cdot))$, 
as well as a representation $\pi:\mathcal{A}\to\mathcal{L}(\mathcal{D})$ 
and a unit vector $\Omega\in\mathcal{D}$, 
such that $\omega=(\Omega,\pi(\cdot)\Omega)$ and $\mathcal{D}=\pi(\mathcal{A})\Omega$.
The {\em GNS triple} $(\mathcal{D},\pi,\Omega)$ is determined up to unitary equivalence.
\end{theo}

\begin{proof}
The first step consists of endowing $\mathcal{A}$ with an inner product 
defined by $(a,b)_\bullet\doteq\omega(a^*b)$ for each $a,b\in\mathcal{A}$.
This is per construction sesquilinear and positive semidefinite, since $(a,a)_\bullet=\omega(a^*a)\geq 0$. 
We need to check Hermiticity, namely that $\overline{(a,b)_\bullet}=(b,a)_\bullet$ holds for all $a,b\in\mathcal{A}$. 
To this avail one needs to take into account the following two identities: 
\begin{equation*}
\begin{array}{rcl}
4a^*b&=&(a+b)^*(a+b)-(a-b)^*(a-b)-i(a+ib)^*(a+ib)+i(a-ib)^*(a-ib);\\
4b^*a&=&(a+b)^*(a+b)-(a-b)^*(a-b)+i(a+ib)^*(a+ib)-i(a-ib)^*(a-ib).
\end{array}
\end{equation*}
The positivity requirement on $\omega$ yields the Cauchy-Schwarz inequality for $(\cdot,\cdot)_\bullet$, 
namely $|(a,b)_\bullet|^2\leq(a,a)_\bullet(b,b)_\bullet$, 
but it does not ensure that non-degeneracy holds for this sesquilinear form. 
Hence, we have to single out the vanishing elements 
introducing the subset $\mathcal{I}=\{a\in\mathcal{A}\;|\;\omega(a^*a)=0\}$. 
This is a closed left ideal of $\mathcal{A}$ (but $\mathcal{I}^*\nsubseteq\mathcal{I}$ in general): 
\begin{itemize}
\item It is a closed subset of $\mathcal{A}$ being the preimage of $0$ 
under the continuous map $a\in\mathcal{A}\mapsto\omega(a^*a)\in\mathbb{R}$;
\item Using the Cauchy-Schwarz inequality, 
we note that $\omega(ba)=0=\omega(a^*b)$ for each $a\in\mathcal{I}$ and for each $b\in\mathcal{A}$. 
In particular, this means that $\mathcal{I}=\{a\in\mathcal{A}\;|\;\omega(ba)=0\;,\forall b\in\mathcal{A}\}$, 
showing that $\mathcal{I}$ is a vector subspace of $\mathcal{A}$;
\item $\omega((ba)^*ba)=\omega((a^*b^*b)a)=0$ for each $a\in\mathcal{I}$ and each $b\in\mathcal{A}$, 
hence $\mathcal{I}$ is a left ideal.
\end{itemize}
We can thus define the vector space $\mathcal{D}\doteq\mathcal{A}/\mathcal{I}$, 
where the latter is the set of equivalence classes $[a]$ 
induced by the following equivalence relation: 
$a\sim a'$ if and only if there exists $b\in\mathcal{I}$ such that $a'=a+b$. 
We can endow $\mathcal{D}$ with the positive definite Hermitian non-degenerate sesquilinear form $(\cdot,\cdot)$ 
defined by $([a],[b])\doteq (a,b)^\prime$ for all $[a],[b]\in\mathcal{D}$, 
where $a$ and $b$ are any representative of the equivalence classes $[a]$ and $[b]$ respectively. 
This is well defined as a consequence of the remarks made above 
and it endows $\mathcal{D}$ with a pre-Hilbert structure. 
Taking the completion of $(\mathcal{D}, (,))$ produces a Hilbert space $\mathcal{H}$. 
The representation $\pi$ can be induced via left multiplication exploiting the fact that $\mathcal{I}$ is a left ideal, 
namely we introduce $\pi:\mathcal{A}\to\mathcal{L}(\mathcal{D})$ 
defined via $\pi(b)[a]=[ba]$ for all $b\in\mathcal{A}$ and for all $[a]\in\mathcal{D}$. 
For each $a,b\in\mathcal{A}$ it is easy to check that $\pi(ab)=\pi(a)\pi(b)$, 
while $\pi(a^*)=\pi(a)^*$ follows from the identity $([a^*b],[c])=([b],[ac])$, 
so that $\pi$ turns out to be a representation. 
Furthermore, setting $\Omega\doteq [\mathbb{I}]$, 
one has $(\Omega,\pi(a)\Omega)=\omega(\mathbb{I}^*a\mathbb{I})=\omega(a)$ for each $a\in\mathcal{A}$ 
and $\pi(\mathcal{A})\Omega=\mathcal{D}$. 
This concludes the identification of the GNS triple. Let us now tackle uniqueness. 
Suppose that one can find another realization of $\omega$ as $(\mathcal{D}^\prime,\pi^\prime,\Omega^\prime)$ 
and let us introduce the operator $U:\mathcal{D}\to\mathcal{D}^\prime$ 
such that $U(\pi(a)\Omega)\doteq\pi^\prime(a)\Omega^\prime$ for each $a\in\mathcal{A}$. 
This is well-defined since $\pi(a)\Omega=0$ means $\omega(a^*a)=0$, 
which yields ${\Vert\pi^\prime(a)\Omega^\prime\Vert^\prime}^2=0$. 
Furthermore $U$ preserves the scalar products, 
namely $(U[a],U[b])^\prime=(\pi^\prime(a)\Omega^\prime,\pi^\prime(b)\Omega^\prime)^\prime
=(\Omega^\prime,\pi^\prime(a^*b)\Omega^\prime)^\prime=\omega(a^*b)=([a],[b])$, 
and has an inverse $U^{-1}:\mathcal{D}^\prime\to\mathcal{D}$, 
defined by $U^{-1}(\pi^\prime(a)\Omega^\prime)\doteq\pi(a)\Omega$ for each $a\in\mathcal{A}$, 
preserving the scalar products as well.
Thus it can be extended to a unitary operator from $\mathcal{H}$ to $\mathcal{H}^\prime$, 
the Hilbert space obtained via completion of $\mathcal{D}^\prime$. 
In other words this means that $\Omega^\prime=U\Omega$ 
and that the defining relation for $U$ can also be read as $\pi'([a])=U\pi(a)U^{-1}$. 
This is nothing but the statement that the two representations $\pi$ and $\pi'$ are unitarily equivalent.
\end{proof}

\subsection{Hadamard states}

The power of the algebraic approach lies in its ability to separate the algebraic relations of quantum fields from the Hilbert space representations of these relations and thus in some sense to treat {\it all} possible Hilbert space representations at once. The other side of this coin is that the definition of an algebraic state we reviewed in the previous subsection is too general and thus further conditions are necessary in order to select the physically meaningful states among all possible ones on {\it e.g.}~the quantum field algebra $\mathcal{F}(M)$.

To this avail it seems reasonable to look at the situation in Minkowski spacetime. Physically interesting states there include the Fock vacuum state and associated multiparticle states as well as coherent states and states describing thermal equilibrium situations. All these states share the same ultraviolet (UV) properties, {\it i.e.}~the same high-energy behaviour, namely they satisfy the so-called {\it Hadamard condition}, which we shall review in a few moments. A closer look at the formulation of quantum field theory in Minkowski spacetime reveals that the Hadamard condition is indeed essential for the mathematical consistency of QFT in Minkowski spacetime, as we would like to briefly explain now. In the following we will only discuss real scalar fields for simplicity. Analyses of Hadamard states for fields of higher spin can be found {\it e.g.}~in Refs. \citen{Dappiaggi:2011cj,Fewster:2003ey,Hollands:1999fc,Sahlmann:2000zr}.

Pointwise products of quantum fields $\widehat{\Phi}(x)$ among themselves such as $\widehat{\Phi}(x)^2$ are not automatically well-defined, so that normal ordering is necessary in order to obtain a well-defined object $\wick{\widehat{\Phi}(x)^2}$. The usual normal-ordering procedure of expanding the quantum field into creation and annihilation operators and rearranging those in the product can be equivalently expressed as
$$\wick{\widehat\Phi^2(x)}\;\doteq \lim_{x\to y}\left(\widehat\Phi(x)\widehat\Phi(y)-\omega_{0,2}(x,y){\mathbb{I}}\right)\,,$$
where $\omega_{0,2}(x,y)\doteq\langle \widehat\Phi(x)\widehat\Phi(y)\rangle_{\omega_0}$ is the two-point correlation function of the field in the Minkowski vacuum state $\omega_0$. In the perturbative treatment of interacting theories products such as $\wick{\widehat\Phi^2(x)}\;\wick{\widehat\Phi^2(y)}$ appear which can be computed by means of Wick's theorem, {\it viz.}
$$\wick{\widehat\Phi^2(x)}\;\wick{\widehat\Phi^2(y)}\;=\;\wick{\widehat\Phi^2(x)\widehat\Phi^2(y)}+4\wick{\widehat\Phi(x)\widehat\Phi(y)}\omega_{0,2}(x,y)+2\left(\omega_{0,2}(x,y)\right)^2\,.$$
In other words, for perturbation theory we need the normal-ordered fields to form an algebra with a product specified by Wick's theorem. However, $\omega_{0,2}(x,y)$ is a singular object as it diverges for lightlike related $x$ and $y$, but it is regular enough to be a (tempered) distribution, {\it i.e.}~one obtains finite numbers if integrating $\omega_{0,2}(x,y)$ with test functions $f(x)$, $h(y)$, and this degree of regularity is sufficient for the mathematical treatment of QFT. Thus for the contractions in the Wick theorem to be well-defined one has to check if pointwise products of $\omega_{0,2}(x,y)$ such as the square $\omega_{0,2}(x,y)^2$ are still regular enough to be distributions. Equivalently, in a momentum-space treatment one has to check whether the momentum space integrals appearing in Wick's theorem converge. The answer to these questions is positive because of the energy positivity property of the Minkowskian vacuum state, and this is the reason why one usually never worries about whether normal ordering is well-defined in quantum field theory on Minkowski spacetime. In more detail, the Fourier decomposition of {\it e.g.}~the massless two-point function 
\begin{equation}\label{eq_MinkowskiMassless}\omega_{0,2}(x,y)=\lim_{\epsilon\downarrow 0}\frac{1}{4\pi^2}\frac1{(x-y)^2+i\epsilon(x_0-y_0)+\epsilon^2}\,,\end{equation}
reads
\begin{equation}\label{eq_FourierMinkowskivaccuum}\omega_{0,2}(x,y)=\lim_{\epsilon\downarrow 0}\frac{1}{(2\pi)^3}\int dk \;\Theta(k_0)\delta(k^2)\;e^{ik(x-y)}e^{-\epsilon k_0}\,,\end{equation}
where $\Theta(k_0)$ denotes the Heaviside step function. We see that the Fourier transform of $\omega_{0,2}$ has only support on the forward lightcone (or the positive mass shell in the massive case). This insight allows to determine (or, rather, to define) the square of $\omega_{0,2}(x,y)$ by a convolution in Fourier space
\begin{align*}
(\omega_{0,2}(x,y))^2 & =\lim_{\epsilon\downarrow 0}\frac{1}{(2\pi)^6}\int dq \int dp \;\Theta(q_0)\;\delta(q^2)\;\Theta(p_0)\;\delta(p^2)\;e^{i(q+p)(x-y)}e^{-\epsilon (p_0+q_0)}\\
& =\lim_{\epsilon\downarrow 0}\frac{1}{(2\pi)^6}\int dk \int dq \;\Theta(q_0)\;\delta(q^2)\;\Theta(k_0-q_0)\;\delta((k-q)^2)\;e^{ik(x-y)}e^{-\epsilon k_0}\,.
\end{align*}
Without going too much into details here, let us observe that the above expression can only give a sensible result (a distribution) if the integral over $q$ converges, {\it i.e.}~if the integrand is rapidly decreasing in $q$. To see that this is the case, note that for an arbitrary but fixed $k$ and large $q$, where here ``large'' is meant in the Euclidean norm on ${\mathbb{R}}^4$, the integrand is vanishing on account of $\delta(q^2)$ and $\Theta(k_0-q_0)$ as $k_0-q_0<0$ for large $q_0$. Loosely speaking, we observe the following: By the form of a convolution, the Fourier transform of $\omega_{0,2}$ is multiplied by the same Fourier transform, but with negative momentum. Since $\omega_{0,2}$ has only Fourier support in one `energy direction', namely the positive one, the intersection of its Fourier support and the same support evaluated with negative momentum is compact, and the convolution therefore well-defined. Moreover, as this statement only relies on the large momentum behaviour of Fourier transforms, it holds equally in the case of massive fields, as the mass shell approaches the light cone for large momenta.

We see that the UV behaviour of the Minkowski vacuum state is vital for the consistency of perturbation theory. However, it is important that other physically reasonable states $\omega$ share these UV properties because, if one wants to compute the expectation value of $:\widehat\Phi^2(x):$ in such a state, one obtains
$$\langle:\widehat\Phi^2(x):\rangle_\omega= \lim_{x\to y}\left(\omega_{2}(x,y)-\omega_{0,2}(x,y)\right)$$
and thus $\omega_{2}(x,y)$ must have the {\it same} singularities $\omega_{2,0}(x,y)$ has in order for the result to make sense. As already mentioned, states such as multi-particle states in Fock space, coherent states and thermal equilibrium states do have this UV property, {\it i.e.}~they satisfy the Hadamard condition. As, {\it e.g.}, thermal states do not have support only for positive energies in momentum space, the Hadamard condition seems to be the natural generalisation of the energy-positivity condition of the Minkowski vacuum state which encodes the UV properties of physical states in QFT.

By now the reader should be convinced that the Hadamard condition is a good selection criterion for physical states in Minkowski spacetime and thus it is indeed also used in quantum field theory in curved spacetime in order both to select physical states among all possible ones and to discuss perturbatively interacting theories. In relation to these concepts, see in particular Refs. \citen{Brunetti:1995rf,Brunetti:2009qc,Hollands:2001nf,Hollands:2001fb}.

We shall now finally review the definition of the Hadamard condition. This property of states can be formulated in two ways, one being a generalisation of the position space form \eqref{eq_MinkowskiMassless} of the vacuum two point function and the other being a generalisation of its momentum space form \eqref{eq_FourierMinkowskivaccuum}. The position space version of the Hadamard condition has been already developed in the seventies in the context of the definition of a regularised stress-energy tensor\cite{Walda, Waldb, Wald2}, recently employed in Refs. \citen{Hollands:2004yh,Moretti:2001qh}, and it is well-suited for actual calculations in particular. For a comparison between different regularization schemes, see Ref. \citen{Hack:2012dm2}. On the other hand, the generalised momentum space version has been developed only in the mid nineties by Radzikowski in two seminal papers Refs. \citen{Radzikowski:1996pa,Radzikowski:1996ei} and it is formulated in terms of {\it microlocal analysis}. While being rather abstract, it is well-suited to tackle and solve conceptual problems and indeed only after Radzikowski's work a full understanding of perturbative interacting quantum field theories in curved spacetimes became possible. 

Following our discussion on the obstructions in the definition of normal ordering, we shall start our review of the Hadamard condition by considering the microlocal aspects of Hadamard states. A standard monograph on microlocal analysis are the books of H\"ormander \citen{Hormander1,Hormander2,Hormander3,Hormander4}, while introductory treatments can be found {\it e.g.}~in Refs. \citen{Brunetti:1999jn,Reed,Strohmeier}. The language of microlocal analysis is necessary because in generic curved spacetimes no sensible coordinate-independent notion of Fourier transform exists. However, the situation in Minkowski spacetime teaches us that not the full Fourier spectrum of a state is important, but only its high-energy limit. Thus, on curved spacetimes, we need to specify this ``high-energy limit of the Fourier spectrum'' in a coordinate-independent manner. Microlocal analysis does just that.

In order to formulate the microlocal version of the Hadamard condition we start by introducing the notion of a wave front set. To motivate it, let us recall that a smooth function  on ${\mathbb R}^m$ with compact support has a rapidly decreasing Fourier transform. If we take a distribution $u$ in $C^{\infty\prime}_0({\mathbb R}^m)$ and multiply it by a function $f\in C^\infty_0({\mathbb R}^m)$ with $f(x_0)\neq 0$, then $uf\in{C^{\infty\prime}}({\mathbb R}^m)$, {\it i.e.}~it is a distribution with compact support. If $fu$ were smooth, then its Fourier transform $\widehat{fu}$ would be smooth and rapidly decreasing. The failure of $fu$ to be smooth in a neighbourhood of $x_0$ can therefore be quantitatively described by the set of directions in Fourier space where $\widehat{fu}$ is not rapidly decreasing. With this in mind, one first defines the wave front set of distributions on ${\mathbb R}^m$ and then extends it to arbitrary manifolds in a second step.

\begin{defi}
\label{def_WaveFrontSet}
A neighbourhood $\Gamma$ of $k_0\in{\mathbb R}^m$ is called {\it conic} if $k\in\Gamma$ implies $\lambda k\in\Gamma$ for all $\lambda\in(0,\infty)$. Let $u\in C^{\infty\prime}_0({\mathbb R}^m)$. A point $(x_0,k_0)\in {\mathbb R}^m\times ({\mathbb R}^m\setminus\{0\})$ is called a regular directed point of $u$ if there exists a function $f\in C^\infty_0({\mathbb R}^m)$ with $f(x_0)\neq 0$ and a conic neighborhood $V$ of $k_0$ such that, for every $n\in{\mathbb N}$, there exists a constant $C_n\in{\mathbb R}$ fulfilling
$$|\widehat{fu}(k)|\leq C_n (1+|k|)^{-n}$$
for all $k\in V$. The {\it wave front set} $WF(u)$ is the complement in ${\mathbb R}^m\times ({\mathbb R}^m\setminus\{0\})$ of the set of all regular directed points of $u$.
\end{defi}

\noindent Let us immediately state a few important properties of wave front sets, the proofs of which can be found in Ref. \citen{Hormander1} (see also Ref. \citen{Strohmeier}).

\begin{theo}
 \label{thm_PropertiesWavefront} Let $u\in C^{\infty\prime}_0({\mathbb R}^m)$.
\begin{itemize}
 \item[a)] If $u$ is smooth, then $WF(u)$ is empty.
 \item[b)] Let $P$ be an arbitrary partial differential operator. It holds $$WF(Pu)\subseteq WF(u)\,.$$
\item[c)] Let $U$, $V$ be open subsets of ${\mathbb R}^m$, let $u\in C^{\infty\prime}_0(V)$, and let $\chi:U\to V$ be a diffeomorphism. The pull-back $\chi^*u$ of $u$, defined by $\chi^*u(f)=u(\chi_*f)$ for all $f\in C^\infty_0(U)$, fulfils
$$WF(\chi^*u)=\chi^*WF(u)\doteq \left\{(\chi^{-1}(x),\chi^*k)\;|\;x\in\chi(U),\;(x,k)\in WF(u)\right\}\,,$$
where $\chi^*k$ denotes the pull-back of $\chi$ in the sense of cotangent vectors. Hence, the wave front set transforms covariantly under diffeomorphisms as a subset of $T^*{\mathbb R}^m$. This allows us to extend its definition to distributions on arbitrary manifolds $M$ by glueing together wave front sets in different coordinate patches of $M$. As a result, for $u\in C^{\infty\prime}_0(M)$, $WF(u)\subseteq T^*M\setminus\{{\mathbf 0}\}$, where ${\mathbf 0}$ denotes the zero section of $T^*M$.
\item[d)] Let $u_1$, $u_2\in C^{\infty\prime}_0(M)$ and let
$$WF(u_1)\oplus WF(u_2)\doteq \left\{(x, k_1+k_2)\;|\;(x,k_1)\in WF(u_1),\; (x,k_2)\in WF(u_2)\right\}\,.$$
If $WF(u_1)\oplus WF(u_2)$ does not intersect the zero section, then one can define the product $u_1u_2$ in such a way that it yields a well-defined distribution in $C^{\infty\prime}_0(M)$ and that it reduces to the standard pointwise product of smooth functions if $u_1$ and $u_2$ are smooth. Moreover, the wave front set of such product is bounded in the following way
$$WF(u_1u_2)\subseteq WF(u_1)\cup WF(u_2)\cup \left(WF(u_1)\oplus WF(u_2)\right)\,.$$
\end{itemize}
\end{theo}

\noindent Note that the wave front set transforms as a subset of the cotangent bundle on account of the covector nature of $k$ in $\exp(ikx)$. The last of the above statements is exactly the criterion for the pointwise multiplication of distributions we have been looking for. Namely, from  \eqref{eq_FourierMinkowskivaccuum} and \eqref{eq_MinkowskiMassless} one can infer that the wave front set of the Minkowskian two-point function (for $m\geq0$) is \cite{Reed}
\begin{equation}\label{eq_MinkwoskiWF}
\begin{array}{rcl}
WF(\omega_2) & = & 
\left\{(x,y,k,-k)\in T^*{\mathbb{R}^8}\;|\; x\neq y,\;(x-y)^2=0,\;k||(x-y),\; k_0>0\right\}\\
& \cup & \left\{(x,x,k,-k)\in T^*{\mathbb{R}^8}\;|\; k^2=0,\; k_0>0\right\}\,,
\end{array}
\end{equation}
where $k||(x-y)$ entails that $k$ is parallel to the vector connecting the points $x$ and $y$. It is the condition $k_0>0$ in particular, which encodes the energy positivity of the Minkowskian vacuum state. We can now rephrase our observation that the pointwise square of $\omega_{0,2}(x,y)$ is a well-defined distribution by noting that $WF(\omega_{0,2})\oplus WF(\omega_{0,2})$ does not contain the zero section. In contrast, we know that the $\delta$-distribution $\delta(x)$ is singular at $x=0$ and that its Fourier transform is a constant. Hence, its wave front set reads
$$WF(\delta)=\{(0,k)\in T^*\mathbb{R}\;|\;k\in{\mathbb R}\setminus\{0\}\}\,,$$
and we see that the $\delta$-distribution does not have a `one-sided' wave front set and, hence, can not be squared. The same holds if we view $\delta$ as a distribution $\delta(x,y)$ on $C^\infty_0({\mathbb R}^2)$. Then
$$WF(\delta(x,y))=\{(x,x,k,-k)\in T^*\mathbb{R}^2\;|\;k\in{\mathbb R}\setminus\{0\}\}\,.$$

The previous discussion suggests that a generalisation of \eqref{eq_MinkwoskiWF} to curved spacetimes is the sensible requirement to select physical states (see also Ref. \citen{Verch:1998zn} for an earlier investigation). We shall now define such a generalisation.

\begin{defi}
 \label{def_HadamardScalar}
Let $\omega$ be a state on the quantum field algebra of the scalar field ${\mathcal F}(M)$. We say that $\omega$ fulfils the {\it Hadamard condition} and is therefore a {\it Hadamard state} if its two-point correlation function $\omega_2$ fulfils
$$WF(\omega_2)=\left\{(x,y,k_x,-k_y)\in T^*{\!M}^2\setminus \{\mathbf 0\}\;|\;
(x,k_x)\sim(y,k_y),\;k_x\triangleright 0\right\}\,.$$
Here, $(x,k_x)\sim(y,k_y)$ implies that there exists a null geodesic $c$ connecting $x$ to $y$ such
that $k_x$ is coparallel and cotangent to $c$ at $x$ and $k_y$ is the parallel transport of $k_x$ from $x$ to $y$ along $c$. Finally, $k_x\triangleright 0$ means that the covector $k_x$ is future-directed.
\end{defi}

In early works this condition has only been required for Gaussian (quasifree) states, {\it i.e.}~states $\omega$ which are completely specified in terms of their two-point function $\omega_2$. However, in Ref. \citen{Sanders:2009sw} it has been shown that this condition is sufficient also for non-Gaussian states, because the singularities of all higher correlation functions are already determined by the singularities of $\omega_2$ and the canonical commutation relations. However, note that certain technical results on the structure of Hadamard states have been proven only for the Gaussian case up to now\cite{Verch:1992eg}.

Having discussed the rather abstract aspect of Hadamard states, let us now turn to their more concrete position space form. To this avail, let us consider a geodesically convex open subset ${\mathcal O}$ of $M$, {\it i.e.}~for any two points $x,y\in{\mathcal O}$ there exists a unique geodesic connecting $x$ and $y$ which lies completely in ${\mathcal O}$. By definition, there are open subsets ${\mathcal O}^\prime_x\subseteq T_xM$ such that the exponential map $\exp_x: {\mathcal O}^\prime_x\to{\mathcal O}$ is well-defined for all $x\in {\mathcal O}$, {\it i.e.}~we can introduce Riemannian normal coordinates on ${\mathcal O}$. For any two points $x$, $y\in {\mathcal O}$, we can therefore define the {\it half squared geodesic distance} $\sigma(x,y)$ as
$$\sigma(x,y)\doteq\frac12 g\left(\exp_x^{-1}(y), \exp_x^{-1}(y)\right)\,.$$
This entity is also called {\em Synge's world function} and is both smooth and symmetric on ${\mathcal O}\times{\mathcal O}$. We now provide the explicit form of Hadamard states.

\begin{defi}
 \label{def_HadamardFormScalar}
Let $\omega_2$ be the two-point function of a state on ${\mathcal F}(M)$, let $t$ be a time function on $(M,g)$, let
$$\sigma_\epsilon(x,y)\doteq\sigma(x,y)+2i\epsilon(t(x)-t(y))+\epsilon^2\,,$$ and let $\lambda$ be an arbitrary length scale. We say that $\omega_2$ is of {\it local Hadamard form} if, for every $x_0\in M$ there exists a geodesically convex neighbourhood ${\mathcal O}$ of $x_0$ such that $\omega_2(x,y)$ on ${\mathcal O}\times {\mathcal O}$ is of the form
\begin{align*}
\omega_2(x,y) & =\lim_{\epsilon\downarrow 0}\frac{1}{8\pi^2}\left(\frac{u(x,y)}{\sigma_\epsilon(x,y)}+v(x,y)\log\left(\frac{\sigma_\epsilon(x,y)}{\lambda^2}\right)+w(x,y)\right)\\
& \doteq\lim_{\epsilon\downarrow 0}\frac{1}{8\pi^2}\left(h_\epsilon(x,y)+w(x,y)\right)\,.
\end{align*}
Here, the {\it Hadamard coefficients} $u$, $v$, and $w$ are smooth, real-valued functions, where $v$ is given by a series expansion in $\sigma$ as
$$v=\sum\limits^\infty_{n=0}v_n \sigma^n$$
with smooth coefficients $v_n$. The bidistribution $h_\epsilon$ is called {\it Hadamard parametrix}, indicating that it solves the Klein-Gordon equation up to smooth terms.
\end{defi}

Note that the above series expansion of $v$ does not necessarily converge on general smooth spacetimes, however, it is known to converge on analytic spacetimes \cite{Garabedian}. One therefore often truncates the series at a finite order $n$ and asks for the $w$ coefficient to be only of regularity $C^n$, see Ref. \citen{Kay:1988mu}. Moreover, the local Hadamard form is a special case of the {\it global Hadamard form} defined for the first time in Ref. \citen{Kay:1988mu}. Such a form assures that there exist no (spacelike) singularities in addition to the lightlike ones visible in the local form. Moreover, the whole concept is independent of the chosen time function $t$. However, as proven by Radzikowski in Ref. \citen{Radzikowski:1996ei} employing the microlocal version of the Hadamard condition, the local Hadamard form already implies the global Hadamard form on account of the fact that $\omega_2$ is positive, its antisymmetric part coincides with the causal propagator $G$ and it fulfils the Klein-Gordon equation in both arguments. It is exactly this last fact which serves to determine the Hadamard coefficients $u$, $v$, and $w$ by a recursive procedure, see {\it e.g.}~Section III.1.2 in Ref. \citen{Hack:2010iw} for a review. It turns out that $u$ and $v$ are determined completely in terms of geometric quantities, while $w$ is the only piece which depends on the quantum state $\omega$. Thus $h_\epsilon(x,y)$ is the universal singular piece which is common to all Hadamard states.

Having discussed the Hadamard form, let us state the already anticipated equivalence result obtained by Radzikowski in Ref. \citen{Radzikowski:1996pa}. See also Ref. \citen{Sahlmann:2000zr} for a slightly different proof, which closes a gap in the proof of Ref. \citen{Radzikowski:1996pa}.

\begin{theo}
 \label{thm_HadamardScalar} Let $\omega_2$ be the two-point function of a state $\omega$ on $\mathcal{F}(M)$. $\omega_2$ fulfils the Hadamard condition of Definition \ref{def_HadamardScalar} if and only if it is of global Hadamard form.
\end{theo}

\noindent By the result of Ref. \citen{Radzikowski:1996ei}, that a state which is locally of Hadamard form is already of global Hadamard form, we can safely replace ``global'' by ``local'' in the above theorem. Moreover, from the above discussion it should be clear that the two-point functions of two Hadamard states differ by a smooth and symmetric function.

Before closing this section and the paper by providing examples and non-examples of Hadamard states, we stress that, in a recent paper \citen{Gerard:2012wb}, it has been outlined a new framework aimed at the construction of Hadamard states via techniques proper of pseudo-differential calculus. The connection between these techniques and those used in many of the examples we propose is an open and interesting problem.

\subsubsection{Examples of Hadamard states}

\begin{itemize}
\item All vacuum states and thermal equilibrium states on ultrastatic spacetimes ({\it i.e.}~spacetimes with a metric $ds^2 = -dt^2 + h_{ij}dx^i dx^j$, with $h_{ij}$ not depending on time) are Hadamard states \cite{Fulling, Sahlmann:2000fh}.
\item Based on the previous statement, it has been proven in in Ref. \citen{Fulling} that Hadamard states exist on {\it any} globally hyperbolic spacetime by means of a spacetime deformation argument.
\item The {\it Bunch-Davies state} on de Sitter spacetime is a Hadamard state\cite{Allen}. It has been shown in Refs. \citen{Dappiaggi:2007mx,Dappiaggi:2008dk} that this result can be generalised to asymptotically de Sitter spacetimes, where distinguished Hadamard states can be constructed by means of a holographic argument; these states are generalisations of the Bunch-Davies state in the sense that the aforementioned holographic construction yields the Bunch-Davies state in de Sitter spacetime.
\item Similar holographic arguments have been used in Refs. \citen{Dappiaggi:2011cj,Dappiaggi:2005ci,Moretti:2005ty,Moretti:2006ks,Moretti:2006mu} to construct distinguished Hadamard states on asymptotically flat spacetimes, to rigorously construct the Unruh state in Schwarzschild spacetimes and to prove that it is Hadamard in Ref. \citen{Dappiaggi:2009fx}, to construct asymptotic vacuum and thermal equilibrium states in certain classes of Friedmann-Robertson-Walker spacetimes in Ref. \citen{Dappiaggi:2010gt} and to construct Hadamard states in bounded regions of any globally hyperbolic spacetime in Ref. \citen{Dappiaggi:2010iq}.
\item A interesting class of Hadamard states in general Friedmann-Robertson-Walker are the {\it states of low energy} constructed in Ref. \citen{Olbermann:2007gn}. These states minimise the energy density integrated in time with a compactly supported weight function and thus loosely speaking minimise the energy in the time interval specified by the support of the weight function. This construction has been generalised to encompass almost equilibrium states in Ref. \citen{Kusku:2008zz} and to expanding spacetimes with less symmetry in Ref. \citen{Them:2013uka}.
\item Hadamard states which possess an approximate local thermal interpretiation have been constructed in Ref. \citen{Schlemmer}. See Ref. \citen{VerchRegensburg} for a review.
\item Given a Hadamard state $\omega$ on the field algebra ${\mathcal F}(M)$ and a smooth solution $\Psi$ of the field equation $P\Psi=0$, one can construct a coherent state by redefining the quantum field $\widehat \Phi(x)$ as $\widehat \Phi(x)\mapsto \widehat \Phi(x) + \Psi(x)\mathbb{I}$. The thus induced coherent state has the two-point function $\omega_{\Psi,2}(x,y)=\omega_{2}(x,y)+\Psi(x)\Psi(y)$, which is Hadamard since $\Psi(x)$ is smooth. Moreover, in Ref. \citen{Sanders:2009sw} it has been proven that given a Hadamard state $\omega$ and an arbitrary element $A$ of the field algebra $\mathcal{F}(M)$, the state obtained by applying $A$ to $\omega$, {\it i.e.}~$\omega_A(\cdot)\doteq \omega(A^*\cdot A)/\omega(A^*A)$ is again Hadamard.
\end{itemize}

\subsubsection{Non-examples of Hadamard states}

\begin{itemize}
\item The so-called {\it $\alpha$-vacua} in de Sitter spacetime\cite{Allen} violate the Hadamard condition as shown in Ref. \citen{Brunetti:2005pr}.
\item Recently a construction for distinguished states (termed {\it Sorkin-Johnston states}) on arbitrary globally hyperbolic spacetimes has been proposed in Ref. \citen{Afshordi:2012jf}. Essentially this construction relies on the spectral decomposition of the causal propagator restricted to a certain time interval of the spacetime. It has been shown in Ref. \citen{Fewster:2012ew} that on ultrastatic spacetimes this construction fails to yield a Hadamard state for almost all time intervals. However, in Ref. \citen{Fewster:2012ew,Afshordi:2012ez} it has been conjectured that a modification of the Sorkin-Johnston construction in the spirit of the states of low energy, {\it i.e.}~with ``smooth interval boundaries'', might lead to Hadamard states.
\item A class of states related to Hadamard states, but in general not Hadamard, is constituted by {\it adiabatic states}. These have been introduced in Ref. \citen{Parker} and put on rigorous grounds by Ref. \citen{Luders}. Effectively, they are states which approximate ground states if the curvature of the background spacetime is only slowly varying. In Ref. \citen{Junker}, the concept of adiabatic states has been generalised to arbitrary curved spacetimes. There, it has also been displayed in a quantitative way how adiabatic states are related to Hadamard states. Namely, an adiabatic state of a specific order $n$ has a certain {\it Sobolev wave front set} (in contrast to the $C^\infty$ wave front set introduced above) and hence, loosely speaking, it differs from a Hadamard state by a biscalar of finite regularity $C^n$. In this sense, Hadamard states are adiabatic states of ``infinite order''.
\end{itemize}

\section*{Acknowledgments}
The work of C.D.\ has been supported partly by the University of Pavia and partly by the Indam-GNFM project {``Influenza della materia quantistica sulle fluttuazioni gravitazionali''}. The work of M.B.\ has been supported partly by a DAAD scholarship. M.B.\ is grateful to the II.\ Institute for Theoretical Physics of the University of Hamburg for the kind hospitality. The work of T.-P. H. is supported by a research fellowship of the Deutsche Forschungsgemeinschaft (DFG).
\vskip 1cm


\end{document}